\documentclass[journal]{IEEEtran}
\IEEEoverridecommandlockouts
\usepackage{cite}
\usepackage{amsmath,amssymb,amsfonts}
\usepackage{algorithmic}
\usepackage{graphicx}
\usepackage{textcomp}
\usepackage{color}
\usepackage{verbatim}
\usepackage{amsthm}
\def\BibTeX{{\rm B\kern-.05em{\sc i\kern-.025em b}\kern-.08em
    T\kern-.1667em\lower.7ex\hbox{E}\kern-.125emX}}
\usepackage{mathtools}
\usepackage{subfigure}
\usepackage{enumitem}  
\graphicspath{{Graphics_pdf/}}
\usepackage[normalem]{ulem}
\usepackage{cases}
\usepackage{fancyhdr}
\usepackage[yyyymmdd,hhmmss]{datetime}
\usepackage{tikz,pgfplots}
\usetikzlibrary{external}
 \pgfplotsset{compat=1.14}

\newtheorem{theorem}{Theorem}
\newtheorem{corollary}{Corollary}
\newtheorem{lemma}{Lemma}

\DeclarePairedDelimiter{\ceil}{\lceil}{\rceil}
\DeclarePairedDelimiter\floor{\lfloor}{\rfloor}
\DeclareMathOperator{\diag}{diag}

\newcommand{\rom}[1]{\uppercase\expandafter{\romannumeral #1\relax}}

\newcommand{\LN}{\bar{N}}

\newcommand{\fiB}{\phi_{\mathrm{B}}}
\newcommand{\fiN}{\phi_{\mathrm{N}}}
\newcommand{\M}{\mathcal{M}}
\newcommand{\Mmin}{n_{\mathrm{B}}}
\newcommand{\Iz}{\mathcal{I}_0}
\newcommand{\IzC}{\mathcal{I}^{\mathrm{c}}_0}
\newcommand{\Io}{{\mathcal{I}}_{0,\mathrm{M}}^{\mathrm{c}}}
\newcommand{\Izo}{{\mathcal{I}}_{0,\mathrm{L}}^{\mathrm{c}}}
\newcommand{\Ioz}{{\mathcal{I}}_{0,\mathrm{R}}^{\mathrm{c}}}

\newcommand{\St}{\mathcal{\tilde{S}}}
\newcommand{\Lzt}{{\tilde{L}_0}}
\newcommand{\fig}{Fig.}    % Style on float: Fig.~1. caption \fullstop
\newcommand{\tab}{Table}  % Style:  TABLE I   \new line  caption \nofullstop
\newcommand{\secR}{Section}

\begin{document}

\title{Hybrid Combining of Directional Antennas for Periodic Broadcast V2V Communication \\
\thanks{This research has been carried out in the antenna systems center \emph{ChaseOn} in a project financed by Swedish Governmental Agency of Innovation Systems (Vinnova), Chalmers, Bluetest, Ericsson, Keysight, RISE, Smarteq, and Volvo Cars.}%
\thanks{The authors are with the Communication Systems Group, Department of Electrical Engineering, Chalmers University of Technology, 412 96 Gothenburg, Sweden (e-mail: chouaib@chalmers.se; erik.strom@chalmers.se; fredrik.brannstrom@chalmers.se)}%
}%

\author{Chouaib~Bencheikh~Lehocine, Erik~G.~Str{\"{o}}m,~\IEEEmembership{Senior Member,~IEEE,} and Fredrik~Br{\"{a}}nnstr{\"{o}}m,~\IEEEmembership{Member,~IEEE}}%
\maketitle

\begin{abstract}
A hybrid analog-digital combiner for broadcast vehicular communication is proposed. It has an analog part that does not require any channel state information or feedback from the receiver, and a digital part that uses maximal ratio combining (MRC). We focus on designing the analog part of the combiner to optimize the received signal strength along all azimuth angles for robust periodic vehicle-to-vehicle (V2V) communication, in a scenario of one dominant component between the communicating vehicles (e.g., highway scenario). We show that the parameters of a previously suggested fully analog combiner solves the optimization problem of the analog part of the proposed hybrid combiner. Assuming $L$ directional antennas with uniform angular separation together with the special case of a two-port receiver, we show that it is optimal to combine groups of $\ceil{L/2}$ and $\floor{L/2}$  antennas in analog domain and feed the output of each group to one digital port. This is shown to be optimal under a sufficient condition on the sidelobes level of the directional antenna. Moreover, we derive a performance bound for the hybrid combiner to guide the choice of antennas needed to meet the reliability requirements of the V2V communication links.

\end{abstract}


%%%%%%%%%%%%%%%%%%			%%%%%%%%%%%%%%%		%%%%%%%%%%%%%%%%%%%%%
\section{Introduction}
%for both vehicular users and pedestrians
\IEEEPARstart{C}{ooperative}~\color{black} intelligent transportation systems (C-ITS) have the promise to make our roads safer, more efficient, and environment-friendly.  %increase traffic efficiency, and reduce the negative impact of transportation systems on environment. 
To enable cooperation between intelligent transportation systems, vehicular communication is needed. In contrast to applications supported by traditional mobile broadband communication, C-ITS applications can have ultra high requirements on reliability and latency, and good designs of antenna systems are key to meet those reliability requirements.
Typically, vehicles are equipped with omnidirectional antenna elements that are mounted on the roof. 
However, the vehicle body and the mounting position of antennas affect the omnidirectional characteristics of the radiation patterns\cite{AntPlac2007,AntPlac_2_2014}. Consequently, the distorted omnidirectional pattern may have very low gains or blind spots at certain azimuth angles. In scenarios where the received signal has only one dominant component, e.g., highway scenarios, problems of decoding the packets may arise if the angle of arrival (AOA) of the signal coincides with a very low gain of the receiving antenna. Hence, we need the antenna system to be robust in all directions of arrivals. 

For a robust broadcast antenna system in spite of body and placement effects, multiple antennas with complementary patterns can be combined to achieve an effective pattern with better omnidirectional characteristics than the distorted patterns of single antennas. In particular, directional antennas that are pointing towards different directions can be used for this purpose. Since the main contribution of these antennas is through the main lobe,  we may have more flexibility---compared to omnidirectional antennas---in mounting them such that the distortions do not coincide with their main lobes. That may reduce the severity of distortions and lead to an overall better effective pattern. 

The classical way of combining antennas is done in the digital domain through either selection combining (SC), equal gain combining (EGC), or maximal ratio combining (MRC). In general, these fully digital solutions require a radio frequency (RF) chain per antenna element which results in high cost and power consumption.

A recent trend in deploying multiple antenna systems is the use of hybrid analog-digital combiners/beamformers. Hybrid combiners (HCs) require less number of RF chains than antennas, which translates to a reduction of cost and power consumption with respect to the fully digital solutions. HCs have been studied extensively and state-of-the-art designs and paradigms are well summarized in~\cite{hybrid,hybrid2,hybrid3}. 
  The main architectures considered use phase shifters in the RF part, which limits the HC to have a constant-modulus complex RF weight. Design strategies of HCs take into account instantaneous or average channel state information (CSI) to attempt finding the optimal analog and digital combining coefficients~\cite{hybrid2}. The design objective is, typically, maximizing the rate or the signal-to-noise ratio (SNR) experienced by the user~\cite{hybrid}. An example of a study of HCs that take into account vehicular users is~\cite{config2019}. There, a general framework is set for the design of hybrid beamformers/combiners according to three main architectures. Performance is evaluated in a cellular infrastructure-to-everything scenario considering vehicular users. 
  
  Aside from the fully digital and hybrid combining schemes, fully analog solutions have been investigated in several publications including~\cite{AC} and~\cite{mainACN}. In~\cite{AC}, three architectures of analog combiners (ACs) have been proposed. The different architectures have varying complexity depending on the combination of RF components used (variable gain amplifiers, variable gain amplifiers and sign switches or variable phase shifters).  
  Optimization of the RF weights in these architectures was done with the objective of maximizing the SNR at the output of the combiner, assuming perfect knowledge of CSI. In~\cite{mainACN}, an analog combiner based on phase shifters was proposed for broadcast vehicle-to-vehicle (V2V) communication. V2V broadcast communication is mainly based on periodic dissemination of status messages. The solution presented in~\cite{mainACN} takes advantage of the periodic behavior of the status messages and the redundancy of their content and attempt to find the optimal phase shifters that minimize the probability of consecutive packet errors in the system. That is done in absence of CSI or feedback from the receiver. Hence, beside the low cost and low power consumption, this AC has very low complexity. However, it does not leverage on the benefit of digital combining.

To take advantage of digital combining high performance benefits and analog combining low-cost, low-complexity benefits, we propose, in this work, a hybrid analog-digital combiner that is tailored for periodic broadcast V2V communication. We take a modular approach in our HC design, that is, we assume that the digital baseband combining is done using a multiport receiver that applies standard MRC. Given that we do not have access to CSI or any feedback from the receiver, we attempt to design the analog part of the HC to minimize the error probability of consecutive status messages. That is a similar design approach to what was followed in~\cite{mainACN}. The proposed HC has the advantage of having the analog part independent of the digital part. Hence, the RF part can be located closer to where the antennas are placed. Moreover, we use directional antennas in this setting and we investigate the optimal HC configuration, i.e., the optimal way to connect antennas to the digital ports, in the special case of a two-port digital receiver. The main contributions of this work can be summarized as follows.
\begin{itemize}
	\item We propose a hybrid low-cost combining scheme for
	broadcast vehicular communication that results in robust performance at all AOAs.  
	\item We show that the obtained parameters for the phase shift network in~\cite{mainACN} are also a solution to the optimization of the analog part of HC when MRC is used in baseband. Depending on the setup, the proposed HC either minimizes or provides a lower bound on the burst error probability of the full system. 
	\item  We analytically identify the optimal number of antennas per port when feeding $L$ directional antennas to the HC with two digital ports. It is found that it is optimal to feed $\ceil{L/2}$ antenna to one port and $\floor{L/2}$ antenna to the other port, when a sufficient condition on the sidelobes of the antennas is satisfied.
	\item In addition, we derive a performance bound of the HC, which can be used as a guideline when designing the antenna system.
\end{itemize}

%%%%%%%%%%%%%%%%%%			%%%%%%%%%%%%%%%		%%%%%%%%%%%%%%%%%%%%%

\section{System Model}\label{S2}
In this section we start by presenting the system model. We note that this is the system model as presented in \cite{mainACN}, extended to multiple ports receiver case.
\subsection{Antenna System}
We assume the use of $L$ antennas that are mounted at the same height on the vehicle. The antennas are vertically polarized \cite{patch}. Taking into account the relative comparable height of different vehicles, we assume that the incident waves arrive in the azimuth plane. Thus, we can represent the far-field response of the antennas as a function of the azimuth angle $\phi$ only. That is,
\begin{equation}
g_l(\phi,\theta)=g_l(\phi, \pi/2)= g_l(\phi),
\end{equation}
where $\theta$ is the elevation angle and $l$ is the antenna index.

\subsection{Channel Model}
We consider a scenario of vehicle-to-vehicle communication in a scarce multi-path (MP) environment with a dominant component. This component could be a line-of-sight (LOS) or a first order single bounce reflection from a local reflector. This is typical in communication between cars in highways. In such scenario, it has been shown through measurements that only a few MP components with small azimuth angular spread  contribute to the total received power \cite{AngSpread2011}. Thus, the channel gain at an antenna output can be modeled by a single dominant component arriving at a particular AOA $\phi$. Taking into account the far-field function of the receiver antennas, the channel gain of antenna $l$  can be stated as \cite[Eq.~8]{karedal2009geometry} %\cite[Eq.~1]{ch2004}
\begin{equation} 
h_l(t)=\tilde{a}(t)g_l(\phi)e^{-\jmath \tilde{\Omega}_{l}(t) },
\end{equation}
where $\tilde{a}(t)$ and $\tilde{\Omega}_{l}(t)$ are the complex gain and the phase shift, respectively, of the dominant MP component with AOA $\phi$, received at antenna $l$. The phase shift depends on the distance as $\tilde{\Omega}_{l}(t)=\frac{2\pi}{\lambda} d_l(t)$, where $\lambda$ is the wavelength of the carrier frequency and $d_l(t)$ is the distance of the path followed from the transmitter to the $l^{\mathrm{th}}$ receiver antenna. Taking the antenna with index $l=0$ as a reference, we can  express the relative phase differences with respect to the reference antenna as
$\Omega_{l}(t)=\tilde{\Omega}_{l}(t)-\tilde{\Omega}_{0}(t)$ with $\Omega_{0}(t)=0$.
The channel model can be expressed consequently as
\begin{equation} \label{channel}
h_l(t)=a(t)g_l(\phi)e^{-\jmath \Omega_{l}(t) },
\end{equation}
with $a(t)=\tilde{a}(t)e^{-\jmath \tilde{\Omega}_{0}(t) }$. The relative phase differences can be assumed slowly varying over a period of a few seconds. This holds if the main component between the two vehicles with approximately the same AOA is still present over that duration. From~\eqref{channel} we can see that in case an antenna $l$ has poor gain at the AOA of the signal, the received power will be low, which could lead to the loss of the received packet. Thus, we aim to design our analog combiner such that the overall scheme is robust for any AOA.

\subsection{IEEE802.11p Broadcast Communication}
Cooperative awareness in intelligent transportation systems that are based on IEEE802.11p, is created through the dissemination of periodic messages that contain status information like position, speed, heading, lane position, etc. These are called cooperative awareness messages (CAMs). CAMs are transmitted periodically with a frequency that is dependent on both the change of status of the transmitting vehicle (e.g., change of lane), and on the traffic on the radio channel~\cite{CAM}. The period between two consecutive CAMs, $T$, has a minimum duration of  $100$ ms and a maximum  of $1000$ ms \cite{CAM}, while the duration of a CAM packet $T_{\mathrm{m}}$ is in the range of $0.5-2$ ms \cite{mainACN}. We note that the period is much longer than the packet duration, $T\gg T_{\mathrm{m}}$.
Due to the high frequency of CAM dissemination, their content is highly correlated over the time of few periods.  At the level of a C-ITS application running on a vehicle, if a CAM packet from a neighboring vehicle is lost, the C-ITS application can still use the previously available status information about the neighboring vehicle and function properly. This is provided that the age of the already available status information is not too large (i.e., the available status information is not outdated). Following this line of thought, it has been proposed in several works including~\cite{Aoinfo}, to assess the reliability of broadcast V2V systems using age-of-information rather than the classical overall packet delivery ratio metric. 
The age-of-information at a certain time instance is equal to the latency of the latest successfully received packet plus the time elapsed since its reception.
We recall that latency of a packet is commonly defined as the time elapsed between the generation and reception of the packet. If we consider the latency as negligible or fixed from packet to packet, the randomness in the age-of-information peak values is determined by the time separating two consecutive successful CAM packet receptions.
We note that packet inter-reception time (acronymed IRT or PIR) and T-window are similar metrics and they are elaborated in~\cite{PIR,PIR0}, and~\cite{Twindow}, respectively.
To relate the application-level age-of-information metric to the reliability of V2V communication, a parameter labeled Burst Error probability (BrEP) has been proposed and used in~\cite{mainACN} to design an analog combiner for V2V communication. Using this parameter, outage is declared when $K$ consecutive CAMs are not decoded correctly.
In our work we follow a similar approach, that is, taking into account the correlated status information (position, speed, heading, lane position, etc.) contained in consecutive CAMs, we assume that only some packets out of a burst of $K$ consecutive CAMs need to be decoded correctly for the proper functionality of a particular C-ITS application that depend on their content.  Therefore, we target the minimization of the BrEP of $K$ CAMs, which if latency is negligible (or constant) is equivalent to minimizing the probability of not satisfying the maximum allowable age-of-information $A_{\text{max}}= KT$ (or $A_{\text{max}}= KT+\text{const.}$) set by the C-ITS application. Aside from reliability, robustness is an important quality of V2V broadcast systems.
To take that into account in our design, we consider a scenario where the AOA remains approximately constant for the time it takes to transmit $K$ packets. In such case, there is a risk of losing all the $K$ packets if the AOA coincides with a blind spot or a notch of low gain of the radiation pattern of the receiving antenna system. Hence, we target to combine the antennas to minimize the probability that a burst of $K$ consecutive packets is received with error when the AOA coincides with the worst-case AOA.

\subsection{Analog Combiner}
Given $L$ antennas, we use an AC
to combine these antenna signals to $P<L$ receiver ports. The digital receiver performs MRC based on the estimated effective channel and noise variance after the analog combining. The AC does not use any CSI or feedback from the receiver. Therefore, we restrict it to be composed of time varying phase shifters and adders. The phase shifters are modeled as a linear function of time. The complex analog gain from antenna $l$ to port $p$ is modeled as
\begin{align}\label{AC}
s_{l,p}e^{\jmath(\alpha_{l,p}t+ \beta_{l,p})}
\end{align}
where $\alpha_{l,p}$ and $\beta_{l,p}$ are the slope and the initial phase offset, respectively, of the time-varying analog phase shifter, and $s_{l,p} \in \mathbb{R} $ is the gain from the output of antenna $l$ to the input of port $p$. For the AC to be power-conserving, it must satisfy 
\begin{equation}
\label{eq:1}
\sum_{p=0}^{P-1} s_{l,p}^2 = 1. 
\end{equation}
\begin{figure*}[ht]
		\centering
	\subfigure[Sub-connected configuration]{
			\includegraphics[width=0.9\columnwidth]{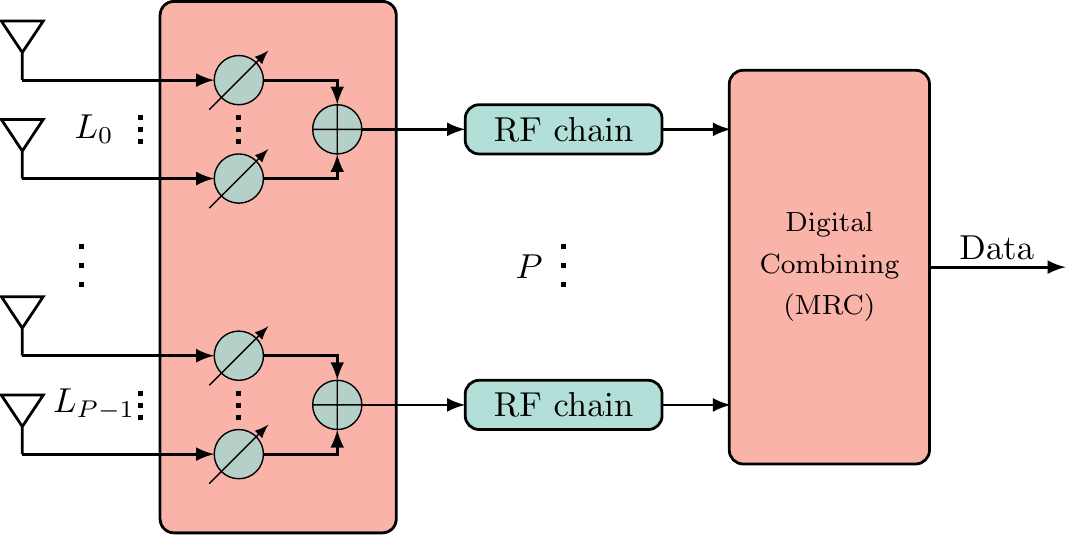}
			\label{Fig:SCHEME:a}
		}
	\subfigure [Fully-connected configuration]
	{
	\includegraphics[width=0.9\columnwidth]{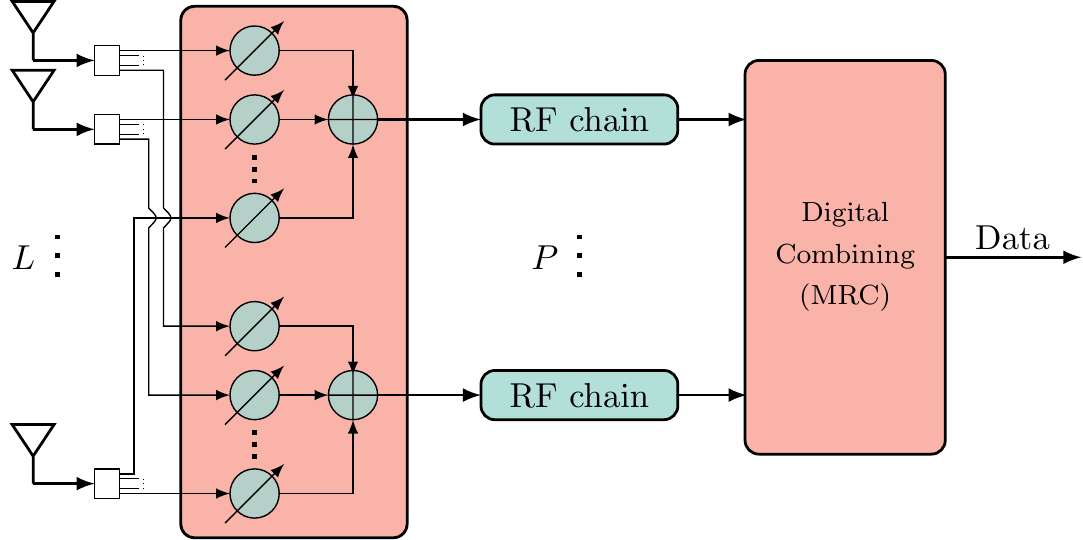}
	\label{Fig:SCHEME:b}
}
\caption{Sub-connected and fully connected architectures of Hybrid combiners.}	
\label{Fig:SCHEME}
\end{figure*}
We define $\mathbf{S}\in\mathbb{R}^{L\times P}$ to be the matrix with elements
$[\mathbf{S}]_{l,p}=s_{l,p}$. Depending on how antennas are connected to RF chains, the HC can have either a sub-connected, a fully-connected or an overlapping configuration~\cite{hybrid}. As shown \fig~\ref{Fig:SCHEME:a}, in the sub-connected configuration
each antenna is connected to only one RF chain.
This results in $P$ disjoint groups, each group comprised of $L_p>0$ antennas and $\sum_{p=0}^{P-1} L_p = L$. The set of sub-connected configurations can be expressed as 
\begin{align}
\label{S}
\mathcal{S} =
\Bigl\{&\mathbf{S}\in\{0,1\}^{L\times P}: \nonumber\\
&\mathbf{S}^\textsf{T}\mathbf{S} = \diag(L_0, L_1, \ldots, L_{P-1}), \sum_{p=0}^{P-1} L_p = L\Bigr\}.
\end{align}
In the fully connected configuration all antennas are connected to all RF chains as shown in \fig~\ref{Fig:SCHEME:b}, while in the overlapping configuration at least one antenna is connected to more than one RF chain. Although the fully connected and overalpping configurations are  more general than the sub-connected configuration, they have the drawback that the noise components at the
receiver ports might be correlated, which complicates the optimal digital
receiver. Moreover, the fully connected and overlapping configurations imply the use of power
splitters, which leads to insertion loss and added complexity. For
these reasons, we chose in this paper to design the HC following a sub-connected configuration.

%%%%%%%%%%%%%%%%%%			%%%%%%%%%%%%%%%		%%%%%%%%%%%%%%%%%%%%%
\section{Design of the Analog Combiner}\label{S3}
\subsection{SNR after the Hybrid Combining}
We start by deriving the SNR expression of a received packet after the analog combining and MRC, where the analog combining is performed in sub-connected fashion with linear time varying phase shifters.
We have that the received signal at the output of antenna $l$ is given by
  \begin{equation}\label{SM:y}
  y_l(t)=h_l(t)x(t) +n_l(t),
  \end{equation}
  where $x(t)$ is the transmitted signal, $h_l(t)$ is the channel gain at the antenna input modeled by~\eqref{channel} and $n_l(t)$ is  independent, zero-mean complex Gaussian and white noise over the signal bandwidth, with average power $\sigma_{\mathrm{n}}^2$. Given the analog combiner complex gain modelled by~\eqref{AC}, the input to receiver port $p$ is
\begin{align}
r_p(t)&=\sum_{l=0}^{L-1}s_{l,p}e^{\jmath(\alpha_{l,p}t + \beta_{l,p})} y_{l}(t) \nonumber\\
&=a(t)x(t)\sum_{l=0}^{L-1}s_{l,p}g_l(\phi)e^{-\jmath\big( \Omega_{l}(t)-\alpha_{l,p}t-\beta_{l,p}\big) }\nonumber\\
&\qquad\qquad+\sum_{l=0}^{L-1}s_{l,p} n_l(t)e^{\jmath(\alpha_{l,p}t+\beta_{l,p} ) }\nonumber\\
&=a(t)x(t) c_p( t)+\tilde{n}_p(t),\label{port}
\end{align}
where $c_p(t)$ and $\tilde{n}_p(t)$ are the equivalent channel gains and noise signals respectively, at the input of receiver port $p$. Since $\tilde{n}_p(t)$  is a sum of phase shifted zero-mean complex Gaussian white noise signals, it has a Gaussian distribution $\mathcal{CN}(0, L_p \sigma_{\mathrm{n}}^2)$, where $L_p=\sum_{l=0}^{L-1} s_{l,p}$ is the number of antennas connected to port $p$.

 We assume that the time variation of the equivalent gain $c_p(t) $ is negligible over the duration of a packet $T_{\mathrm{m}}$, and therefore we make the approximation
\begin{align}\label{capp}
c_p( t) &\approx c_p( kT),~
 kT\leq t\leq kT+T_{\mathrm{m}},~k=0,\ldots , K-1.
\end{align}
Moreover, as stated before, in our scenario of interest the relative phase difference $\Omega_{l}(t)$ can be assumed to experience negligible variation over a duration $KT$~\cite{mainACN} and thus, it can be approximated as  $\Omega_{l}(t)\approx \Omega_{l}\triangleq \Omega_{l}(0)$. Furthermore, we define $\psi_{l,p}\triangleq \mod(\Omega_l-\beta_{l,p}-\angle(g_l(\phi)),2\pi) $ to accommodate for the effective phase response at antenna $l$
 and we write $c_p$ adopting the previous approximations as,
\begin{equation} \label{arg2}
\begin{split}
c_p( kT)&=\sum_{l=0}^{L-1}s_{l,p}|g_{l}(\phi)|e^{-\jmath( \psi_{l,p}-\alpha_{l,p}kT)}.\\
\end{split}
\end{equation} 
The receiver uses a MRC weight vector to combine the signals. In the sub-connected configuration there is no correlation between the noise processes at the input of the receiver, and therefore the optimal MRC weight applied at port $p$ is $c_p^*(kT)/L_p$~\cite{MRC2002}.
Following that, the SNR of the $k^{\mathrm{th}}$ packet after the MRC can be found to be 
\begin{equation}\label{snr}
\gamma_k =\frac{P_\mathrm{r}}{\sigma^{2}_\mathrm{n}}\sum_{p=0}^{P-1} \frac{|c_p( kT)|^2}{L_p}, 
\end{equation}
where $P_\mathrm{r}=\mathbb{E}\{|a(t)x(t)|^2\}$ is the average signal power, and it is assumed to remain constant over the period of $K$ consecutive packets.

We would like to express $\gamma_k$ as a function of the AC parameters that need to be optimized. We note that only phase slopes $\alpha_{l, p}$ of the phase shifters that are connected to a port, i.e., for the pairs $(l, p)$ such that $s_{l, p} = 1$, need to be considered in the optimization problem.
Therefore,  we define the vectors $\boldsymbol{\alpha}_p\in\mathbb{R}^{L_p}$, $p=0, 1, \ldots, P-1$ where
\begin{align}
\{ [\boldsymbol{\alpha}_p]_i: i=0, \ldots, L_p-1\}=\{&\alpha_{l,p}:  s_{l,p}=1, \nonumber\\
 &\quad l=0, ..., L-1\}.
\end{align}
We assume that the mapping between the elements of the sets is such that it assigns $i=0$ in the first set to the smallest $l$ in the second set, and $i=1$ to the second smallest $l$ and so on. We also define the vectors $\boldsymbol{\psi}_p$ as
$\{ [\boldsymbol{\psi}_p]_i, i=0, ..., L_p-1        \}=\{ \psi_{l,p}:  s_{l,p}=1, l=0, ..., L-1\} $, where the mapping from $i$ to $l$ is similar to the mapping used when defining the vectors $\boldsymbol{\alpha }_p$. Moreover, we define $L$-element vectors, $\boldsymbol{\alpha}$ and $\boldsymbol{\psi}$ as $ \boldsymbol{\alpha}=[\boldsymbol{\alpha}_0^\textsf{T}, \boldsymbol{\alpha}_1^\textsf{T}, ..., \boldsymbol{\alpha}_{P-1}^\textsf{T}]^\textsf{T}$ and  $ \boldsymbol{\psi}=[\boldsymbol{\psi}_0^\textsf{T}, \boldsymbol{\psi}_1^\textsf{T}, ..., \boldsymbol{\psi}_{P-1}^\textsf{T}]^\textsf{T}$. Then, we can readily express the SNR in~\eqref{snr} as 
$\gamma_k(\phi, \mathbf{S}, \boldsymbol{\psi},\boldsymbol{\alpha})$.

\subsection{Optimization Problem of the Analog Combiner}
To formulate the optimization problem of the analog part of the HC, the burst error probability needs to be derived. Given a burst of $K$ packets, under the assumptions that packets are of the same length, are transmitted using the same modulation and coding scheme, and packet errors are statistically independent, the BrEP  is
\begin{equation}\label{PB}
P_{\mathrm{B}}(\phi, \mathbf{S}, \boldsymbol{\psi}, \boldsymbol{\alpha})=\prod_{k=0}^{K-1}P_{\mathrm{e}}\big(\gamma_k(\phi, \mathbf{S}, \boldsymbol{\psi}, \boldsymbol{\alpha})\big),
\end{equation}
where $P_{\mathrm{e}}(\gamma_k(\phi, \mathbf{S}, \boldsymbol{\psi},\boldsymbol{\alpha}))$ is the packet error probability (PEP) of the $k^\mathrm{th}$ packet.
Since we are interested in robustness, the design goal is to
minimize the bust-error probability for the worst-case propagation,
i.e., when all power arrives in the least favorable AOA,
and the worst-case effective phase response at the antennas outputs,
$\boldsymbol{\psi}$. Design variables are the grouping configuration and phase
slopes defined by $\mathbf{S}$ and
$\boldsymbol{\alpha}$, respectively.
Hence, the optimal
analog combiner is determined by $\mathbf{S}^\star$ and
$\boldsymbol{\alpha}^\star$, where 
\begin{equation}
\label{eq:3}
(\mathbf{S}^\star, \boldsymbol{\alpha}^\star)
\triangleq \arg
\inf_{\substack{\mathbf{S}\in\mathcal{S}\\ \boldsymbol{\alpha}\in\mathbb{R}^{L}}}
\sup_{\substack{\phi\in[0, 2\pi)\\ \boldsymbol{\psi}\in[0,
		2\pi)^{L}}}
P_{\mathrm{B}}(\phi, \mathbf{S}, \boldsymbol{\psi}, \boldsymbol{\alpha}).
\end{equation}
We will start by finding the optimal phase slopes for a given
configuration $\mathbf{S}$, i.e.,
\begin{equation}\label{phases}
%\label{eq:4}
\boldsymbol{\alpha}^\star(\mathbf{S})
= \arg
\inf_{\boldsymbol{\alpha}\in\mathbb{R}^{L}}
\sup_{\substack{\phi\in[0, 2\pi)\\ \boldsymbol{\psi}\in[0,
		2\pi)^{L}}}
P_{\mathrm{B}}(\phi, \mathbf{S}, \boldsymbol{\psi}, \boldsymbol{\alpha}),
\end{equation}
and then find the optimal configuration as
\begin{equation}\label{config}
\mathbf{S}^\star
= \arg
\inf_{\mathbf{S}\in\mathcal{S}}
\sup_{\substack{\phi\in[0, 2\pi)\\ \boldsymbol{\psi}\in[0,
		2\pi)^{L}}}
P_{\mathrm{B}}(\phi, \mathbf{S},  \boldsymbol{\psi}, \boldsymbol{\alpha}^\star(\mathbf{S})).
\end{equation}

\subsection{Optimal Phase Slopes}
The choice of the phase slopes $\alpha_{l,p}$ applied to the antennas, is made according to~\eqref{phases}. %where we minimize the error probability for the worst possible AOA $\phi$, and  the worst initial phase shifts $\boldsymbol{\psi}$.
The solutions to this optimization problem depends on the PEP function. For the special case of exponential PEP function defined as
\begin{equation}\label{burst}
P_{\mathrm{e}}(\gamma)=a\exp(-b\gamma),
\end{equation}
where $a,b>0$ are constants, the BrEP can be expressed as
\begin{equation}\label{PB2}
P_{\mathrm{B}}(\phi, \mathbf{S}, \boldsymbol{\psi}, \boldsymbol{\alpha})=a^{K}\exp\big(-b\sum_{k=0}^{K-1}\gamma_k(\phi, \mathbf{S}, \boldsymbol{\psi}, \boldsymbol{\alpha})\big).
\end{equation}
To solve~\eqref{phases}, we first formulate a problem to find the optimal phase slopes for any AOA, that is $\boldsymbol{\alpha}^\star(\phi, \mathbf{S})$. Then, we can deduce the optimal phase slopes for the worst-case AOA. We have,
\begin{align}\label{subphases}
\boldsymbol{\alpha}^\star(\phi,\mathbf{S})&=\arg\inf_{\boldsymbol{\alpha}\in \mathbb{R}^{L}}
\sup_{\boldsymbol{\psi}\in[0,
		2\pi)^{L}}
P_{\mathrm{B}}(\phi, \mathbf{S}, \boldsymbol{\psi}, \boldsymbol{\alpha }).
\end{align}
Introducing the logarithm function to~\eqref{PB2} and substituting in the previous equation, we get 
\begin{align}
\boldsymbol{\alpha}^\star(\phi,\mathbf{S})&=\arg
\inf_{\boldsymbol{\alpha}\in \mathbb{R}^{L}}
\sup_{\boldsymbol{\psi}\in[0,
		2\pi)^{L}}
\ln\big(P_{\mathrm{B}}(\phi, \mathbf{S}, \boldsymbol{\psi}, \boldsymbol{\alpha })\big)\nonumber\\
&= \arg
\sup_{\boldsymbol{\alpha}\in \mathbb{R}^{L}}
\inf_{\boldsymbol{\psi}\in[0,
		2\pi)^{L}}
 \sum_{k=0}^{K-1}\gamma(\phi,\mathbf{S},\boldsymbol{\psi},\boldsymbol{\alpha},k)\nonumber\\
 &= \arg
 \sup_{\boldsymbol{\alpha}\in \mathbb{R}^{L}}
 \inf_{\boldsymbol{\psi}\in[0,
 		2\pi)^{L}}
\overbrace{ \sum_{p=0}^{P-1}
 \underbrace{\sum_{k=0}^{K-1} \frac{|c_p(kT    )|^2}{L_p}}_{= J_p(\phi, \mathbf{S}, \boldsymbol{\psi}_p,\boldsymbol{\alpha}_p)} }^{=J(\phi, \mathbf{S}, \boldsymbol{\psi}, \boldsymbol{\alpha })}.\label{obj1}
\end{align}

We note that the elements $\boldsymbol{\alpha }_p$ that compose the vector $\boldsymbol{\alpha } $ are independent of each other. The same thing applies to $\boldsymbol{\psi}$. Following that, we can decompose the joint optimization of $J(\phi, \mathbf{S}, \boldsymbol{\psi}, \boldsymbol{\alpha })$ to the separate optimization of the terms $J_p(\phi, \mathbf{S}, \boldsymbol{\psi}_p,\boldsymbol{\alpha}_p)$.  After separate optimization of these terms, we can readily reconstruct the optimal phase slopes vector for the overall system as $\boldsymbol{\alpha}^\star(\phi,\mathbf{S})=[\boldsymbol{\alpha}^\star_{ 0}(\phi,\mathbf{S}), \boldsymbol{\alpha}^\star_{ 1}(\phi,\mathbf{S}), ..., \boldsymbol{\alpha}^\star_{ P-1}(\phi,\mathbf{S})]^\textsf{T}$ where  
\begin{equation}\label{JP}
\boldsymbol{\alpha}^\star_{p}(\phi, \mathbf{S})
= \arg
\sup_{\boldsymbol{\alpha}_p \in \mathbb{R}^{L_p}}
\inf_{ \boldsymbol{\psi}_p\in[0,
		2\pi)^{L_p}}
J_p(\phi, \mathbf{S}, \boldsymbol{\psi}_p,\boldsymbol{\alpha}_p).
\end{equation}
The problem in~\eqref{JP}, i.e., the optimization of an analog combiner for a single port receiver, has been tackled in\cite{mainACN}. The solutions to this problem are stated in the following theorem.
\begin{theorem}\label{Th1}[adopted from \cite{mainACN}] The optimum of the objective function $J_{p}$ defined in~\eqref{obj1} for a given configuration $\mathbf{S}$, is lower bounded for any AOA $\phi$ and $p=0, 1, \ldots, P-1$ as,
	\begin{equation}\label{optim}
	\begin{split}
	J^\star_{p}(\phi, \mathbf{S}) 
	&\geq K\frac{1}{L_p}\sum_{l=0}^{L-1}s_{l,p}|g_{l}(\phi)|^2,\qquad L_p \leq K,
	\end{split}
	\end{equation}
	this lower bound is achievable when selecting 
	\begin{align*}
	\boldsymbol{\alpha}_{p}= \tilde{\boldsymbol{\alpha}}_{p} (\mathbf{S}),
	\end{align*}
	where 
	\begin{align}\label{opt_ph}
		[\tilde{\boldsymbol{\alpha}}_{p}(\mathbf{S}) ]_i=
		\frac{i2\pi}{KT}, \quad  i=0, 1, \ldots, L_p-1.  
	\end{align}
	 For the special case of  $L_p\leq3$, this lower bound is tight, thus 
	\begin{equation}\label{optim2}
	J^\star_{p}(\phi, \mathbf{S})= K\frac{1}{L_p}\sum_{l=0}^{L-1}s_{l,p}|g_{l}(\phi)|^2,\quad  L_p \leq K \text{ and } L_p\leq 3,
	\end{equation}
	and
	\begin{align}\label{eq:phase_ind}
	\boldsymbol{\alpha}^\star_{p}(\phi,\mathbf{S})= \tilde{\boldsymbol{\alpha}}_p(\mathbf{S}), \quad  L_p \leq K \text{ and } L_p\leq 3.
	\end{align}
	
\end{theorem}
\begin{proof}
	See \cite[Appendix]{mainACN}.
	\end{proof}

Following the results stated in Theorem~\ref{Th1}, we can deduce that the optimum of the objective function  defined in~\eqref{obj1} $J^\star(\phi, \mathbf{S})$ is bounded as
\begin{align} \label{eq:LB}
J^\star(\phi, \mathbf{S})&\geq K\sum_{p=0}^{P-1}\frac{1}{L_p}\sum_{l=0}^{L-1}s_{l,p}|g_{l}(\phi)|^2, &L_p \leq K ,\forall p.
\end{align}
The bound is achievable, when selecting a phase slopes vector 
$\tilde{\boldsymbol{\alpha}}(\mathbf{S})= 
[\tilde{\boldsymbol{\alpha}}_{ 0}^\textsf{T}(\mathbf{S}), \tilde{\boldsymbol{\alpha}}_{ 1}^\textsf{T}(\mathbf{S}), \ldots, \tilde{\boldsymbol{\alpha}}_{ P-1}^\textsf{T}(\mathbf{S})]^\textsf{T}$. That is, when $L_p \leq K$, $\forall p$,
\begin{align}\label{eq:J_alpha_tilde}
	\inf_{\boldsymbol{\psi}\in[0,2\pi)^{L}} J(\phi, \mathbf{S}, \boldsymbol{\psi}, \tilde{\boldsymbol{\alpha}} )&=K\sum_{p=0}^{P-1}\frac{1}{L_p}\sum_{l=0}^{L-1}s_{l,p}|g_{l}(\phi)|^2 \\
	&=K\overline{ G}(\phi,\mathbf{S}),
\end{align}
where 
\begin{align} \label{G}
\overline{ G}(\phi,\mathbf{S})&=\sum_{p=0}^{P-1} \frac{1}{L_p}\sum_{l=0}^{L-1}s_{l,p}|g_{l}(\phi)|^2, \quad \mathbf{S}\in \mathcal{S}.
\end{align}
We will refer to $\overline{ G}(\phi,\mathbf{S})$ as the equivalent radiation pattern after the hybrid combiner.
 In the case $L_p\leq 3$, $\forall p$ the bound in~\eqref{eq:LB} is tight, thus
 \begin{align}
 \boldsymbol{\alpha}^\star(\phi,\mathbf{S})=	\tilde{\boldsymbol{\alpha}}(\mathbf{S}),  \quad L_p \leq K, L_p\leq 3, \forall p.
 \end{align}
Since the optimal phase slopes vectors are independent of $\phi$, we can conclude that the solution of~\eqref{subphases} is equivalent to the solution of~\eqref{phases}, i.e., $\boldsymbol{\alpha}^\star(\phi,\mathbf{S})=\boldsymbol{\alpha}^\star(\mathbf{S})$ for all $\phi$, including the worst-case AOA.

\subsection{Optimal Sub-connected Configuration}
 In~\eqref{config} we stated the problem of finding the optimal sub-connected configuration $\mathbf{S}^\star$ that corresponds to the optimal phase slopes $\boldsymbol{\alpha}^\star(\mathbf{S})$, defined according to~\eqref{phases}. In the past subsection we found the phase slopes $\tilde{\boldsymbol{\alpha}}(\mathbf{S})$ that ensures an upper bound on the BrEP for the configurations $\{\mathbf{S}\in \mathcal{S}: L_p\leq K, \forall p\}$. For the configurations in the subset $\{\mathbf{S}\in \mathcal{S}: L_p\leq K$ and $L_p\leq 3,\forall p \}$, the phase slopes are optimal, that is, $\tilde{\boldsymbol{\alpha}}(\mathbf{S})=\boldsymbol{\alpha}^\star(\mathbf{S})$. However, for the configurations $\{\mathbf{S}\in \mathcal{S}: \exists p, L_p>K\}$, no analytical solution is available for phase slopes that minimizes or achieves an upper bound on BrEP. Taking that into account, let us define a subset of $\mathcal{S}$, 
 \begin{align}
 	\St=\{\mathbf{S}\in \mathcal{S}: L_p\leq K, \forall p\} 
 \end{align}
 together with an optimization problem to find the optimal configuration for the hybrid combiner with an analog part that uses phase shifters with slopes $\tilde{\boldsymbol{\alpha}}(\mathbf{S})$, that is 
\begin{align}\label{def:obj:Stilde}
\tilde{\mathbf{S}}
&\triangleq\arg
\inf_{\mathbf{S}\in\St}
\sup_{\substack{\phi\in[0, 2\pi)\\ \boldsymbol{\psi}\in[0,
		2\pi)^{L}}}
P_{\mathrm{B}}(\phi, \mathbf{S},  \boldsymbol{\psi}, \tilde{\boldsymbol{\alpha}}(\mathbf{S})).
\end{align}
The solution to~\eqref{config} for the configurations $\{\mathbf{S}\in \mathcal{S}: L_p\leq K$ and $L_p\leq 3,\forall p \}$ is a special case of~\eqref{def:obj:Stilde}. We follow similar steps as used to derive~\eqref{obj1} to express~\eqref{def:obj:Stilde} as
\begin{align}
\tilde{\mathbf{S}}	&=\arg
\sup_{\mathbf{S}\in\St}
\inf_{\substack{\phi\in[0, 2\pi)\\ \boldsymbol{\psi}\in[0,
		2\pi)^{L}}}
J(\phi, \mathbf{S}, \boldsymbol{\psi}, \tilde{\boldsymbol{\alpha}} )\\
&=\arg
\sup_{\mathbf{S}\in\St}
\inf_{\phi\in[0, 2\pi)}
K\overline{ G}(\phi,\mathbf{S}), \label{S*} 
\end{align}
where~\eqref{S*} follows from~\eqref{eq:J_alpha_tilde}.
The optimal configuration $\tilde{\mathbf{S}}$ depends on the radiation patterns of the combined antennas. Given a set of antennas, we pick the configuration that satisfies~\eqref{S*}. That is the configuration that maximizes the sum of the SNRs from the $K$ packets for the worst-case AOA $\phi$. Since the BrEP is inversely proportional to the equivalent gain $\overline{ G}(\phi,\mathbf{S})$ for all $\phi$, the overall system performance can be assessed according to the gain of the worst-case AOA of the equivalent radiation pattern after hybrid combining, $\min_{\phi }\overline{ G}(\phi,\mathbf{S})$. 
We observe that $\tilde{\mathbf{S}}$ maximizes the bound on $J^\star(\phi,\mathbf{S})$ in~\eqref{eq:LB}. Hence, the hybrid combiner with the configuration $\tilde{\mathbf{S}}$ gives the best upper bound on the BrEP at the worst-case AOA. We identify two special cases of~\eqref{def:obj:Stilde}, 
\begin{itemize}
	\item If $\mathcal{S}=\St=\{\mathbf{S}\in \mathcal{S}: L_p\leq K$ and $L_p\leq 3,\forall p \}$, then the solution to~\eqref{eq:3} is given by $(\mathbf{S}^\star, \boldsymbol{\alpha}^\star)=(\tilde{\mathbf{S}}, \tilde{\boldsymbol{\alpha}})$. That is the optimal sub-connected analog combiner that minimizes the BrEP for the worst AOA when MRC is used for digital combining.
	\item If $\St=\emptyset $ then we need to solve~\eqref{eq:3} numerically to attempt to find phase slopes and configuration that minimizes or achieves an upper bound on the BrEP.
\end{itemize}
Assuming that $\St\neq\emptyset $, we expect high performance, i.e., a high gain at the worst-case AOA of  the equivalent radiation pattern after the HC with an analog part defined by ($\tilde{\mathbf{S}}$, $\boldsymbol{\tilde{\alpha} }  (\mathbf{S})$) and MRC combining. To assess the limits of the hybrid combiner, a performance bound is derived in the following lemma. 
\begin{lemma}\label{lem:limits}
	Let  $\overline{ G}(\phi,\mathbf{S})$ be defined according to~\eqref{G} and let the average radiation of all antenna elements be the same in the azimuth plane, that is
	\begin{align}\label{lem1:ass}
	\frac{1}{2\pi}\int_{2\pi}|g_l(\phi)|^2 d\phi =G_{\mathrm{2}\pi}, \quad \forall l. 
	\end{align}
	Then, 
	\begin{equation}\label{lem:lim1}
	\inf_{\phi\in[0, 2\pi)}
	\overline{ G}(\phi,\tilde{\mathbf{S}}) \leq P G_{\mathrm{2}\pi},
	\end{equation}
	where $\tilde{\mathbf{S}}$ is defined according to~\eqref{def:obj:Stilde}.
	\begin{proof}
		See Appendix \ref{appendix_Lem1}.
	\end{proof}
\end{lemma}
An example of a scenario where~\eqref{lem1:ass} holds is the deployment of identical directional antennas pointing at different directions along the azimuth plane.
We can observe from the right-hand side of~\eqref{lem:lim1} that the bound on the performance is proportional both to the average gain of the antenna elements and the number of ports. However, it is  not dependent on the number of antennas $L$. 
In \secR~\ref{NR}, it will be shown through some examples, that the performance bound can be approached with limited number of directional antennas, which implies that the conditions associated with the use of $\tilde{\boldsymbol{\alpha}}$, ($\St\neq\emptyset $) do not set a limitation on the use of the HC.

%%%%%%%%%%%%%%%%%%			%%%%%%%%%%%%%%%		%%%%%%%%%%%%%%%%%%%%%

\section{Hybrid Combing of Directional Antennas }\label{S4}
For a generic far field functions $g_l(\phi)$, it is hard to analytically solve the problem stated in~\eqref{S*}. This however, can be solved given some characterization of the far field function of antennas. In the following we assume the use of $L$ identical directional antennas elements that are evenly spread around the azimuth plane. We attempt to solve~\eqref{S*} analytically for the HC with analog phase shifters slopes $\tilde{\boldsymbol{\alpha}}(\mathbf{S})$ and a digital MRC receiver with $P=2$. Let $G(\phi)=|g(\phi)|^2$ be the radiation pattern of the antenna in the azimuth plane centered around $\phi=0$, where $g(\phi)$ is the far field function of the antenna. Then, the radiation pattern of antenna $l$ is given by
\begin{equation} \label{Gl}
	G_l(\phi)= G(\phi-l2\pi/L),\quad \phi \in [-\pi,\pi).
\end{equation}
We assume that the radiation pattern $G(\phi)$ is symmetric\footnote{The assumption is needed only to simplify the definition of the main lobe region and sidelobes region of the pattern. If the assumption is dropped, we need minor modifications on the definition of the main lobe and sidelobes regions of the pattern for the results of this section to hold.} around $\phi=0$ with $G(0)=\max_{\phi} G(\phi)$. 
Let  $G_{\mathrm{SL}}$ be the gain of the highest sidelobes such that the sidelobes level (SLL) is obtained as $10\log_{10} \big(G_{\mathrm{SL}}/G(0)\big) $. We define $\fiB$ as
\begin{align} \label{phi_B}
	\fiB=\min \{\phi \in [0,\pi): G(\phi)=G_{\mathrm{SL}} \}.
\end{align}
Following this, $G(\phi)$ can be lower bounded as
\begin{align} \label{B_ml}
	&G(\phi) > G_{\mathrm{SL}}, \quad \phi \in (-\fiB,\fiB).
\end{align}
We refer to this region of the pattern where  $\phi \in (-\fiB,\fiB)$ as the main lobe of the pattern.  
The remaining region is referred to as the sidelobes region of the pattern, in which $G(\phi)$ is upper bounded as
 \begin{align}\label{B_sl}
 	G(\phi) \leq G_{\mathrm{SL}}, \quad \phi \in [-\pi, -\fiB]\cup [\fiB,\pi).
 \end{align} 
  Also, we refer to $\fiB$  as the break point of the radiation pattern, as it separates the main lobe region of the pattern from the sidelobes region. 
  Due to the directional behavior of antennas, we assume that\footnote{This upper bound is used to simplify the proof of Theorem~\ref{th2}.}  $\fiB<\pi/2$. 
  Then, we attempt to find the configuration $\tilde{\mathbf{S}}$ 
that maximizes the gain at the worst-case AOA of the effective radiation pattern after the analog-digital combining, that is  $\overline{ G}(\phi,\mathbf{S})$. We consider sub-connected configurations that group adjacent antennas together. We group $L_0$ consecutive antennas together and the remaining consecutive $L_1=L-L_0$ antennas together. In such setting $\mathbf{S}$ can be determined from $L_0$ and therefore we write, with some abuse of notation, $\overline{G}(\phi,L_0)$ instead of $\overline{ G}(\phi,\mathbf{S})$.
 Without loss of generality, we let  $L/2\leq L_0\leq L-1$ and given $P=2$, the equivalent radiation pattern obtained from~\eqref{G} is 
\begin{equation}\label{G_eq}
\begin{split}
	\overline{ G}(\phi,L_0)
	&=\sum_{l=0}^{L_0-1}\frac{1}{L_0}G_l(\phi)+\sum_{l=L_0}^{L-1}\frac{1}{L-L_0}G_l(\phi).
\end{split}
\end{equation}
 Due to the adjacent grouping of antennas, $\mathbf{S} \in \St$ is equivalent to $L_0\in \{l \in \mathbb{N}: L/2\leq l\leq L-1$ and $l\leq K\}$. Thus,~\eqref{S*} can be expressed only in terms of $L_0$ as
\begin{equation}\label{L1_opt}
	\tilde{L}_0=\arg\sup_{\substack{\lceil L/2 \rceil\leq L_0\leq L-1\\ \quad ~~L_0\leq K}}%\\\\substack{quad ~~L_0\leq K}}
	\quad \inf_{\phi\in[0, 2\pi)}
	\overline{ G}(\phi,L_0) , \quad \tilde{L}_1=L-\tilde{L}_0.
\end{equation}
 Note that $(\tilde{\mathbf{S}})^\textsf{T}\tilde{\mathbf{S}}=\diag (\tilde{L}_0,L-\tilde{L}_0)$, according to~\eqref{S}.
The solution to a relaxed version of this problem is formulated in the following theorem. 
\begin{theorem} \label{th2} 
    Suppose $P=2$. Consider $L$ antennas with radiation patterns $G_l(\phi) = G(\phi- 2\pi/L)$ for $l=0, 1, \ldots, L-1$, where $G(\phi)$ is such that $G(\phi)=G(-\phi)$, $G(0)=\sup_{\phi} G(\phi)$, and has break point $\fiB<\pi/2$, where $\fiB$ is defined in~\eqref{phi_B}. 
    If the gain of the highest sidelobe $G_{\mathrm{SL}}$ of $G(\phi)$ satisfies 
	\begin{align} \label{condition}
    G_{\mathrm{SL}}< \frac{\floor{L/2}}{\ceil{L/2}}\inf_{\phi \in  \mathcal{J}}	\frac{1}{L-1}\sum_{l=0}^{L-2}G_l(\phi),
	\end{align}
	where $\mathcal{J}=[\fiB-\frac{2\pi}{L}, (L-1)\frac{2\pi}{L}-\fiB]$, then (i) $\fiB> \pi/L$ and (ii) the solution 
	to~\eqref{L1_opt}, with the constraint $L_0\leq K$ relaxed, is
    \begin{equation}\label{L1*}
        \tilde{L}_0= \ceil{L/2}, \quad \tilde{L}_1 =\floor {L/2} .
    \end{equation}

\begin{proof}
	See Appendix \ref{appendix_Th2}.
\end{proof}
\end{theorem}

\begin{corollary}\label{cor:1}
Given the assumptions in Theorem~\ref{th2}, if 
\begin{equation}\label{c_c}
G_{\mathrm{SL}}< \frac{\floor{L/2}}{\ceil{L/2}}\inf_{\phi \in  [0,2\pi/L)}	\frac{1}{L}\sum_{l=0}^{L-1}G_l(\phi)
\end{equation} 
is satisfied, then~\eqref{condition} holds. 
\begin{proof}
	See Appendix \ref{appendix_Cor1}.
\end{proof}
\end{corollary}
Theorem~\ref{th2} gives the solution to a relaxed version of~\eqref{L1_opt}, under the fulfillment of~\eqref{condition}.
The condition requires a computation of only one infimum, which is simpler than solving the problem stated in~\eqref{L1_opt}.
 Condition~\eqref{c_c} presented in Corollary~\ref{cor:1} is stricter than~\eqref{condition}, but requires a computation of the infimum over a simpler interval. 
  Both~\eqref{condition} and~\eqref{c_c}  are sufficient but not necessary conditions for the optimality of~\eqref{L1*},  and therefore if neither of them holds then~\eqref{L1*} may be still the optimal solution of the relaxed version of~\eqref{L1_opt}, which can be solved numerically in that case.
Another result of the theorem is that~\eqref{condition} implies $\fiB> \pi/L$. Since we are interested in a reliable system, antennas with $\fiB\leq \pi/L$ are not a suitable choice in the first place. The main lobes of such antennas do not intersect, which results in low performance in the directions covered by only sidelobes, in spite of the chosen configuration.
 Hence we can infer that~\eqref{condition} omits the choice of such antennas that lead to low performance.
 
Given that~\eqref{condition} holds and assuming that $\ceil{L/2}\leq K$, we can conclude that the configuration that has grouping of antennas according to~\eqref{L1*} is feasible in~\eqref{L1_opt} and thus optimal. The HC with  analog part given by $(\tilde{\mathbf{S}},\tilde{\boldsymbol{\alpha}}(\mathbf{S}))$, where $\tilde{\mathbf{S}}$ has $\tilde{L}_0=\lceil L/2 \rceil$, achieves the best lower bound (with respects to bounds achieved using other configurations) on the gain of the equivalent radiation pattern at the worst-case AOA. 
As done in the previous section, special cases of~\eqref{L1_opt} are identified as follows. 
\begin{itemize}
		\item If $L-1 \leq K$ and $L-1 \leq 3$, then $\mathcal{S}=\St=\{\mathbf{S}\in \mathcal{S}: L_p\leq K$ and $L_p\leq 3,\forall p \}$, where $\mathcal{S}$ is restricted to adjacent configurations. Therefore, the solution to~\eqref{eq:3} is given by $\mathbf{S}^\star$ with $L_0^\star=\tilde{L}_0=\ceil{L/2}$, and $ \boldsymbol{\alpha}^\star=\tilde{\boldsymbol{\alpha}}(\mathbf{S})$.
        This minimizes the BrEP for the worst-case AOA among the hybrid combiners with adjacent configurations. 
       \item If $\ceil{L/2}> K$, then $\St=\emptyset$. Thus, ~\eqref{eq:3} needs to be solved numerically. 
\end{itemize} 
Assuming that $\ceil{L/2}\leq K $, (thus excluding the case when $\ceil{L/2} > K$) and given that the analog combiner is defined by the configuration $\tilde{\mathbf{S}}$ that has $(\ceil{L/2}, \floor{L/2}) $ and the phase slopes $\tilde{\boldsymbol{\alpha}}(\mathbf{S})$, we use Lemma~\ref{lem:limits} to deduce that 
	\begin{equation}
	\inf_{\phi\in[0, 2\pi)}
	\overline{ G}(\phi,\tilde{L}_0) \leq 2G_{\mathrm{2}\pi}.
	\end{equation}
This is the performance bound when deploying the HC with $P=2$ and $L$ identical directional antennas with equidistant angular separation. The bound is the same for both adjacent and non-adjacent configurations. We make the observation that when $L$ is even, the optimal adjacent configuration with $\tilde{L}_0=\lceil L/2 \rceil=L/2$ has the same equivalent radiation pattern $\overline{G}(\phi,L/2)=L/2\sum_{l=0}^{L-1}G_l(\phi)$ as non-adjacent configurations with $L_0=L/2$. Moreover, we will show later in the numerical results section that performance bound can be approached when $L$ is even, with configuration that has an adjacent setting. Therefore adjacent grouping of antennas is not suboptimal. From another aspect, we will illustrate that only a limited number of antennas is needed to approach this bound. Thus, the assumption $\ceil{L/2}\leq K $ does not limit the use of the designed HC. 
\section{Numerical Results}\label{NR}
\subsection{Testing $\eqref{condition}$ and $\eqref{c_c}$ with Patterns of Varying SLL }\label{VA}
\begin{table}[t]
		\caption{\textsc{Normalized far field function of analytical directional radiation patterns\cite[Table~2]{adp18}}}
	\centering
	\begin{tabular} {||c|c|c|c||}
		\hline
		Pattern & Far field function: $g(\phi)$  & $a$ & SLL [dB]\\  
		\hline
		ADP1& $\frac{\sin(a \phi)}{a\phi}$ &$a=\frac{\pi}{\fiN}$ &$-13.2$ \\[3ex] 
		\hline
		ADP2 & $(\frac{\pi}{2})^2\frac{\cos(a\phi)}{(\frac{\pi}{2})^2-(a\phi)^2}$ &$a=\frac{3\pi}{2 \fiN}$ &$-23$\\[3ex] 
		\hline
		ADP3&$\frac{\pi^2}{a\phi}\frac{\sin(a \phi)}{(\pi)^2-(a \phi)^2}$&$a=\frac{2\pi}{\fiN}$ &$-31.5$\\[3ex] 
		\hline
	\end{tabular}
	\label{T1}
\end{table}
In this section we would like to assess if conditions~\eqref{condition} and~\eqref{c_c} presented in Theorem~\ref{th2} and Corollary~\ref{cor:1}, respectively,  can be easily met when using some standard analytical directional patterns (ADPs). We keep the same settings as in \secR~\ref{S4}. We employ $L$ identical directional antennas that have equidistant angular separation in the azimuth plane satisfying~\eqref{Gl}. We test the system when using three ADPs that are characterized by different SLLs. These patterns correspond to the far field radiation of an aperture antenna fed by different shapes of field distribution\cite{adp18}. The far field functions of the ADPs are given in \tab~\ref{T1}. They are normalized with respect to their maxima and expressed as a function of $\fiN$, where this is defined as the first null of the pattern, $\fiN=\min  \{ \phi \in [0,\pi) : G(\phi)=0 \}$. The break point defined in~\eqref{phi_B} satisfies $\fiB<\fiN$. The three ADPs with $\fiN=65 $\,$\deg$ are shown in \fig~\ref{ADP0}.

Given $L$ and the antenna pattern characteristics defined by both the choice of ADP and $\fiN$, we compute the following parameters
\begin{align}
	U_{1} = \frac{\floor{L/2}}{\ceil{L/2}}
	&\quad\inf_{\phi \in \mathcal{J}}\quad	\frac{1}{L-1}\sum_{l=0}^{L-2}G_l(\phi), \label{C1db}\\
	U_{2}= \frac{\floor{L/2}}{\ceil{L/2}} &\inf_{\phi \in  [0,2\pi/L)}	\frac{1}{L}\sum_{l=0}^{L-1}G_l(\phi).\label{C2db}
\end{align}
Based on these, we can conclude that~\eqref{condition} is satisfied for the considered system ($L$, ADP, and $\fiN$) if $G_{\mathrm{SL}} < U_1$ and~\eqref{c_c} is satisfied if $G_{\mathrm{SL}} < U_2$.

 In \fig~\ref{ADP1}, we have the results of computing $U_1$ for systems with different number of ADP$1$ antenna elements. We also plot a reference line of $G_{\mathrm{SL}}=-13.2 $\,dB, which corresponds to the SLL level of ADP1 because the pattern is normalized. The intersection points of the curves with the reference line indicates the ADP1 with the minimal $\fiN$ such that~\eqref{condition} is satisfied. 
 For $L=4$ all ADP1 antennas with $\fiN \ge 59 $\,$\deg$ satisfies~\eqref{condition}.  Then, with the increase of $L$, $U_1$ increases, hence more patterns satisfy ~\eqref{condition}, e.g., when $L=10$, all ADP1 antennas with $\fiN>28$\,$\deg$ satisfy the condition. 
\begin{figure}[t]
	\center
	\includegraphics[width=\columnwidth]{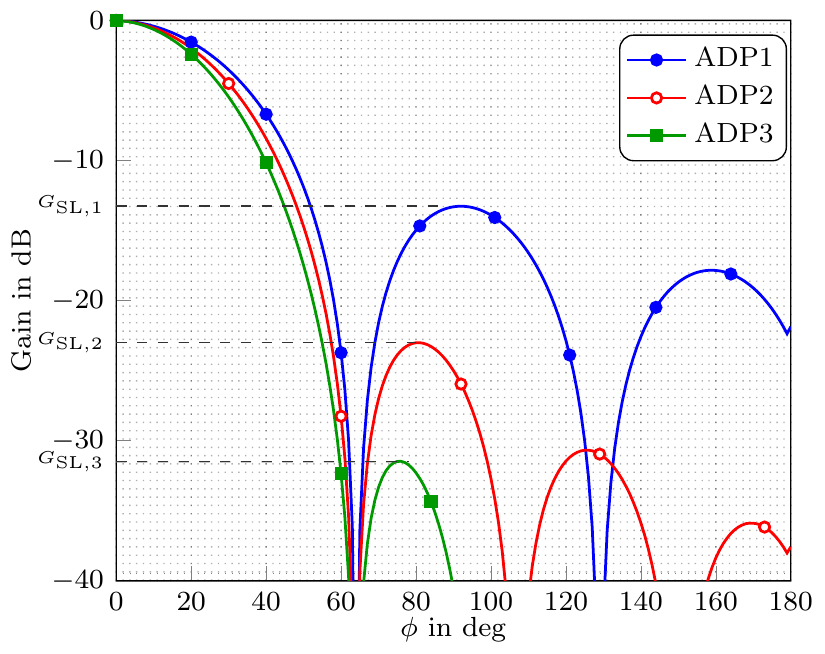}
	\caption{Normalized radiation patterns corresponding to ADPs $1$, $2$, and $3$, with $\fiN=65\deg $.}	
	\label{ADP0}
\end{figure}
We recall from \secR~\ref{S4}, that when~\eqref{condition} is satisfied, the optimal configuration for the HC with $P=2$, assuming $\ceil{L/2}\leq K$,  has $\Lzt=\ceil{L/2}$. To check if $\Lzt=\ceil{L/2}$, even when~\eqref{condition} is not satisfied, we need to solve~\eqref{L1_opt}. We do that numerically when deploying $L=4$ and $L=6$  antennas with ADP1. In \fig~\ref{ADPF} we plot the infimum of the equivalent radiation pattern $\inf_{\phi } \overline{G}(\phi,L_0)$ for different $L_0$ as a function of the first null of the pattern $\fiN$. We can see that the curve corresponding to $\inf_{\phi} \overline{G}(\phi,\ceil{L/2})$ is superior to the curves corresponding to $\inf_{\phi} \overline{G}(\phi,L_0), L_0 \neq \ceil{L/2}$, for both cases $L=4$ and $L=6$. These results can be verified to be valid for other systems with $L\le 10$ and they indicate that $\Lzt=\ceil{L/2}$ is optimal when using any ADP1 antenna with $\fiN\leq 90$, and $L\leq 10$. We highlight in the same figure, the regions where Theorem~\ref{th2} condition~\eqref{condition} applies for the case $L=4$ and $L=6$. While it can be easily checked that inside this region, the performance of the system is at least above $G_{\mathrm{SL}}$, we observe that outside it, the performance of most antennas system is low, since it is below $G_{\mathrm{SL}}$. 
Hence, when~\eqref{condition} does not apply, the antenna system is likely to have poor performance. This relates to claim (i) of Theorem~\ref{th2}. In \tab~\ref{T2}, we show the minimal $\fiN$ and the corresponding break point $\fiB$ for the three ADPs, such that~\eqref{condition} is fulfilled. We note that for all $L$, $\fiB>\pi/L$, which is consistent with claim~(i) of Theorem~\ref{th2}. As discussed in the previous section, this indicates that~\eqref{condition} cuts out the choice of antennas that are likely to lead to poor performance (e.g., antennas with $\fiB\leq \pi/L$). 

From another aspect, we can see in \tab~\ref{T2} that~\eqref{condition} is satisfied for $L=4$ elements of ADP2 and ADP3 with $\fiN\ge52 $\,$\deg$ and $\fiN\ge 50 $\,$\deg$, respectively. While for ADP1 it is satisfied when $\fiN\ge59$\,$\deg$. This indicates that the lower the SLL of the ADP is, the narrower the main lobe of the antenna pattern can be, such that~\eqref{condition} is satisfied. The same observation can be made for the other values of $L$.
\begin{figure}[t]
	\center
	\includegraphics[width=\columnwidth]{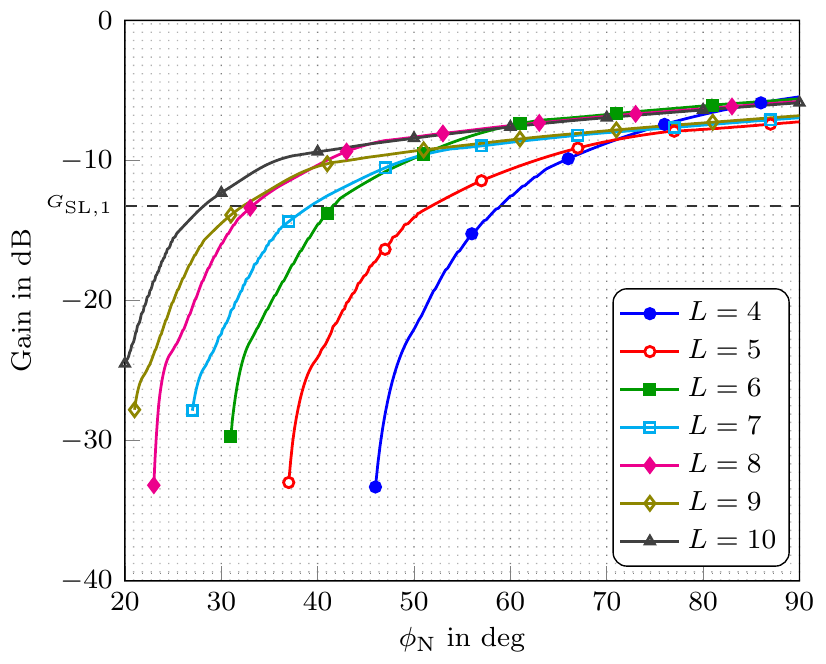}
	\caption{Plot of $U_1$ as a function of ADP1 first null, $\fiN$, for different $L$.}	
	\label{ADP1}
\end{figure}
\begin{figure}[t]
	\center
	\includegraphics[width=\columnwidth]{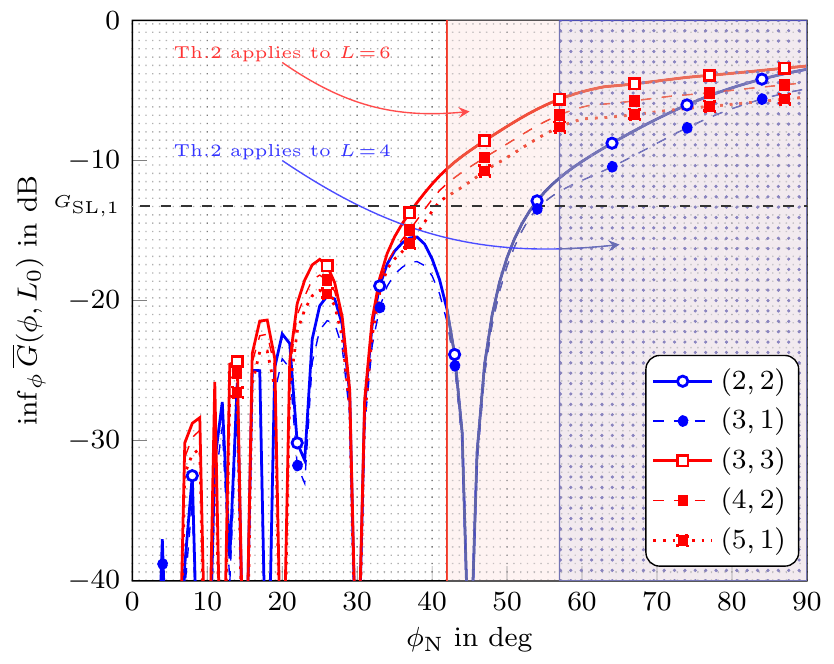}
	\caption{Curves of the gain of the equivalent radiation pattern at the worst-case AOA as a function of ADP1 first null $\fiN$ for different configurations $(L_0,L_1)$ with $L=4,6$.}	
	\label{ADPF}
\end{figure}
\begin{table}[t]
	\centering
	\caption{\textsc{Minimal $\fiN~(\deg)$ and $\fiB~(\deg)$  of ADPs 1, 2 and 3 satisfying~\eqref{condition}}}
	\begin{tabular} {c c|c|c|c|c|c|c|}
		\cline{3-8}
		& &     \multicolumn{2}{ |c| }{ADP1}&\multicolumn{2}{ |c| }{ADP2} &\multicolumn{2}{ |c| }{ADP3} \\
		\hline
		\multicolumn{1}{ ||c| }{$L$}& \multicolumn{1}{ |c| }{$180/L$ }&$\fiN $ & $\fiB$ & $\fiN$ & $\fiB $&$\fiN$ & $\fiB$ \\  
		\hline
		\multicolumn{1}{ ||c| }{$4$}& \multicolumn{1}{ |c| }{$45$}&	$59$  & $48 $  									& $52$  & $47 $  &  $50$  & $46$\\  
		\hline
		\multicolumn{1}{ ||c| }{$5$}&  \multicolumn{1}{ |c| }{$36$}&	$52$  & $42 $  									&$ 43$ &$39$ &  $41$  & $38$ \\
		\hline
		\multicolumn{1}{ ||c| }{$6$}& \multicolumn{1}{ |c| }{$30$}&  $42$  & $34 $																		& $ 36$ &$32$ &  $34$  & $31$  \\
		\hline
		\multicolumn{1}{ ||c| }{$7$}&  \multicolumn{1}{ |c| }{$25.7$}&    $40$  & $32 $																		& $ 32 $ &$28$  &  $30$  & $28$ \\
		\hline
		\multicolumn{1}{ ||c| }{$8$}& \multicolumn{1}{ |c| }{$22.5$}&    $34$  & $27 $																		& $ 28$ &$25$ &  $26$  & $24$ \\
		\hline
		\multicolumn{1}{ ||c| }{$9$}&  \multicolumn{1}{ |c| }{$20$}& $33$  & $27 $																		& $ 25$ &$23$ &  $23$  & $22$  \\
		\hline
		\multicolumn{1}{ ||c| }{$10$}&  \multicolumn{1}{ |c| }{$18$}&  $29$  & $23 $																		& $ 23$ &$20$  &  $21$  & $20$\\
		\hline
	\end{tabular}
	\label{T2}
\end{table}

To check the result of Corollary~\ref{cor:1}, we plot the curves corresponding to $U_1$  and $U_2$ in \fig~\ref{ADP2}. These are computed based on ADP2, for $L=4, 5$, and $6$. As expected, the curves corresponding to $U_2$  are below the curves corresponding to $U_1$,  in the region where~\eqref{condition} holds, since~\eqref{c_c} is more restrictive than~\eqref{condition} in that region. The gap between the two curves is minor, hence we can conclude that~\eqref{c_c} well approximates~\eqref{condition}.

The results shown in this subsection demonstrates that the conditions~\eqref{condition} of Theorem~\ref{th2} and~\eqref{c_c} of Corollary~\ref{cor:1}  are easily satisfied for the three analytical directional pattern models with varying SLL levels. Hence, we can conclude that the SLL required to satisfy the two conditions is within reasonable numbers. Also, we expect the results of Theorem~\ref{th2} to be valid for antennas that have similar pattern shape as the three ADPs. Moreover, if the conditions do not apply, then the antenna system is likely to have poor performance.
\begin{figure}[t]
	\center
	\includegraphics[width=\columnwidth]{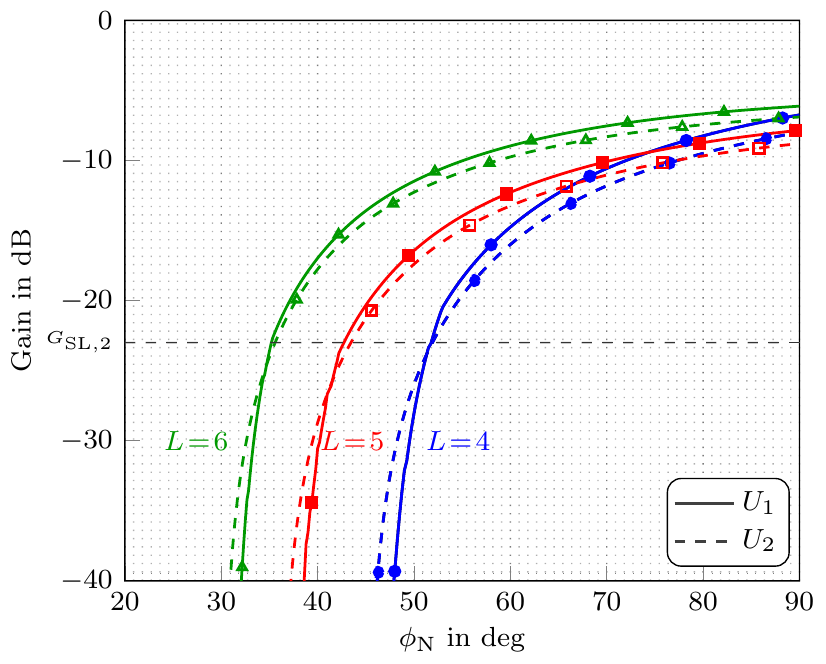}
	\caption{$U_1$ (solid) and $U_2$ (dashed) curves computed as a function of ADP2 first null $\fiN$, for $L=4, 5$, and $6$.}	
	\label{ADP2}
\end{figure}
\begin{figure}[t]
	\center
		\includegraphics[width=\columnwidth]{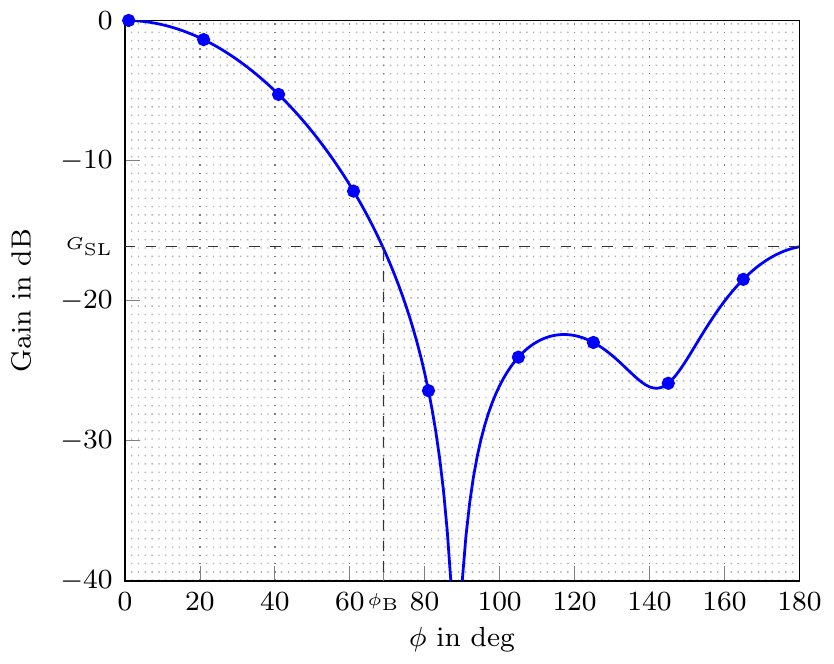}  
	\caption{Normalized radiation pattern of a patch directional antenna at frequency $f_{\mathrm{c}}=5.9$ GHz with length = $0.461\lambda/\sqrt{\epsilon_{\mathrm{r}}}$, width =  $0.512\lambda/\sqrt{\epsilon_{\mathrm{r}}}$ and height =  $\lambda/50$, where $\epsilon_{\mathrm{r}}=1.030$.}	
	\label{pattern_dB}
\end{figure}

\subsection{Performance of the Hybrid Combiner in V2V Scenario}\label{S5B}
We consider $K=10$ and we use that throughout this section. We recall that $K$ refers to the length of a burst of CAM packets, that are periodically broadcasted each $T$\,s. Given $L$ identical directional antennas and $P=2$, the hybrid combining is done using an analog combiner with phase shifters slopes $\tilde{\boldsymbol{\alpha}}$ as defined in \secR~\ref{S3}, followed by MRC. Since $K=10$, for $\St$ to be non-empty,  $\ceil{L/2}\leq K$. This implies that the maximum number of antennas that can be used is $\max(L)=20$. 

Our objective is to minimize the BrEP of the worst-case AOA for a burst of $K$ packets. As explained before, this can be stated as maximizing the effective gain of the worst-case AOA and therefore this will be used as our main metric in comparing the different configurations. 
Apart from the gain of the equivalent radiation pattern at the worst-case AOA, we would like to assess the omnidirectional characteristics of the resultant pattern. This is done based on the difference between the highest and lowest gain of the equivalent pattern in dB, 
\begin{align}
	R(L_0)= 10\log_{10} \frac{\sup_{\phi} G(\phi,L_0) }{\inf_{\phi} G(\phi,L_0)}.
\end{align}
 We use a patch antenna for our computations. These are low-profile antennas that are characterized by directional patterns. They are cheap to manufacture and can be integrated into multiple positions on the vehicle\cite[Ch. 4.1]{patch}. Hence, they are suitable for vehicular multiple antenna systems. The normalized analytical pattern of the patch antenna used is shown in \fig~\ref{pattern_dB}. The normalization is performed with respect to the maximum directive gain of  $9.21$\,dBi. The highest sidelobe gain of the normalized pattern is $G_{\mathrm{SL}}=-16.15$\,dB and the corresponding break point is $\phi_{\mathrm{B}}\approx 69$\,$\deg$.

We start by testing the performance of a system with $L=6$. The antennas are distributed around the azimuth plan with uniform separation of $ 2\pi/L$ according to~\eqref{Gl}. To know which configuration is optimal, we use the results of Theorem~\ref{th2} or Corollary~\ref{cor:1}. 
We check here both conditions~\eqref{condition} and~\eqref{c_c} by computing $U_1$ and $U_2$, which are given by~\eqref{C1db} and~\eqref{C2db}, respectively. We get that $U_1=-6.86$\,dB and $U_2= -7.63$\,dB. 
Since $G_{\mathrm{SL}}=-16.15$\,dB $< U_2 <U_1$, then we conclude that the optimal configuration has $(\Lzt,\tilde{L}_1)=(3,3)$.
It can be represented as $\tilde{\mathbf {S}}=[\mathbf{s}_0,\mathbf{s}_1]$ with
$\mathbf{s}_0=[1, 1, 1, 0, 0,0]^\textsf{T}$ and $\mathbf{s}_1=[0, 0, 0, 1, 1, 1]^\textsf{T}$.  The analog part of the HC is realized using
\begin{align*}
	\tilde{\boldsymbol{\alpha}}_0(\tilde{\mathbf {S}})=[\tilde{\alpha}_{0,0},\tilde{\alpha}_{1,0},\tilde{\alpha}_{2,0}]^\textsf{T}=[0, \frac{2\pi}{KT}, \frac{4\pi}{KT}]^\textsf{T},\\
	\tilde{\boldsymbol{\alpha}}_1(\tilde{\mathbf {S}})=[\tilde{\alpha}_{3,1},\tilde{\alpha}_{4,1},\tilde{\alpha}_{5,1}]^\textsf{T}=[0, \frac{2\pi}{KT}, \frac{4\pi}{KT}]^\textsf{T}.
\end{align*} 

In \fig~\ref{comparison_L6} the equivalent radiation pattern for the different configurations is shown.
Indeed, $(L_0,L_1)=(3,3)$
maximizes the gain at the worst-case AOA and achieves the best performance. Also, the resultant radiation pattern is characterized by good omnidirectional characteristics with $R(\Lzt)=0.51$\,dB. 
Generally, performance decreases as the difference between the number of antennas fed to the two ports $(L_0-L_1)$ increases. 
The equivalent gain of configurations with $L_0\neq L_1$ exhibits beam-forming effect in the directions of antennas scaled by $1/L_1$ (i.e., antennas with indices $l=L_0+1, \ldots, L-1$) which results in lower gain in the directions of the other group of antennas that is scaled by $1/L_0$. 
 Since we are considering broadcast communication between the vehicles, such directional effective radiation pattern is not beneficial. 
 Hence, having $L_0=\Lzt=L/2$ is the best way to divide the available gain evenly in all directions. 
In case $L$ is odd we cannot divide the antennas evenly among the two ports. We give an example using $L=5$. 
We have $U_1= -8.88$\,dB, $U_2= -9.72$\,dB, and $G_{\mathrm{SL}}=-16.15 $\,dB $< U_2 <U_1$. Thus, $\Lzt=3$ and the corresponding optimal scheme has $(\Lzt, \tilde{L}_1)=(3,2)$.  This configuration can be expressed in matrix form as $\tilde{\mathbf {S}}=[\mathbf{s}_0,\mathbf{s}_1]$ with
$\mathbf{s}_0=[1, 1, 1, 0, 0]^\textsf{T}$ and $\mathbf{s}_1=[0, 0, 0, 1, 1]^\textsf{T}$. The analog combiner is realized using phase shifters with slopes 
\begin{align*}
	\tilde{\boldsymbol{\alpha}}_0(\tilde{\mathbf {S}})&=[\tilde{\alpha}_{0,0},\tilde{\alpha}_{1,0},\tilde{\alpha}_{2,0}]^\textsf{T}=[0, \frac{2\pi}{KT}, \frac{4\pi}{KT}]^\textsf{T},\\
\tilde{\boldsymbol{\alpha}}_1(\tilde{\mathbf {S}})&=[\tilde{\alpha}_{3,1},\tilde{\alpha}_{4,1}]^\textsf{T}=[0, \frac{2\pi}{KT}]^\textsf{T}.
\end{align*}

The equivalent radiation pattern for configurations with $L=5$ is shown in \fig~\ref{comparison_L6}. It exhibits beam-forming effect and it has $R(\Lzt)=2.77$\,dB. Despite that the optimal configuration has superior performance than the configuration with $(L_0,L_1)=(4,1)$, nevertheless it does not result in good omnidirectional properties. 
Thus, for best omnidirectional characteristics $L$ should be even so that antennas can be grouped evenly among the two ports.  
\begin{figure}[t]
	\center
	\includegraphics[width=\columnwidth]{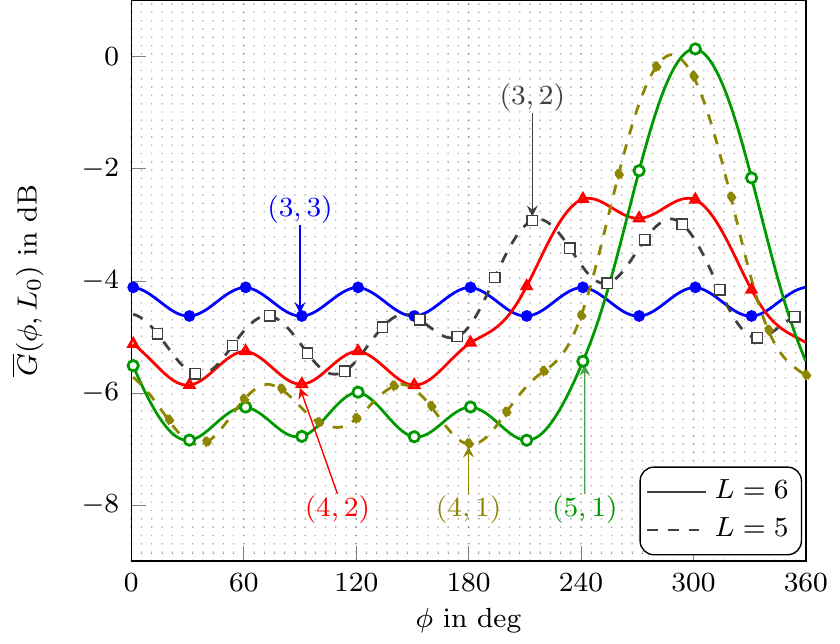}  
	\caption{Effective radiation pattern for hybrid combiner for different configurations $(L_0,L_1)$ with $L=5$, $6$ patch antennas and $K=10$.}	
	\label{comparison_L6}
\end{figure}
\begin{figure}[t]
	\center
	\includegraphics[width=\columnwidth]{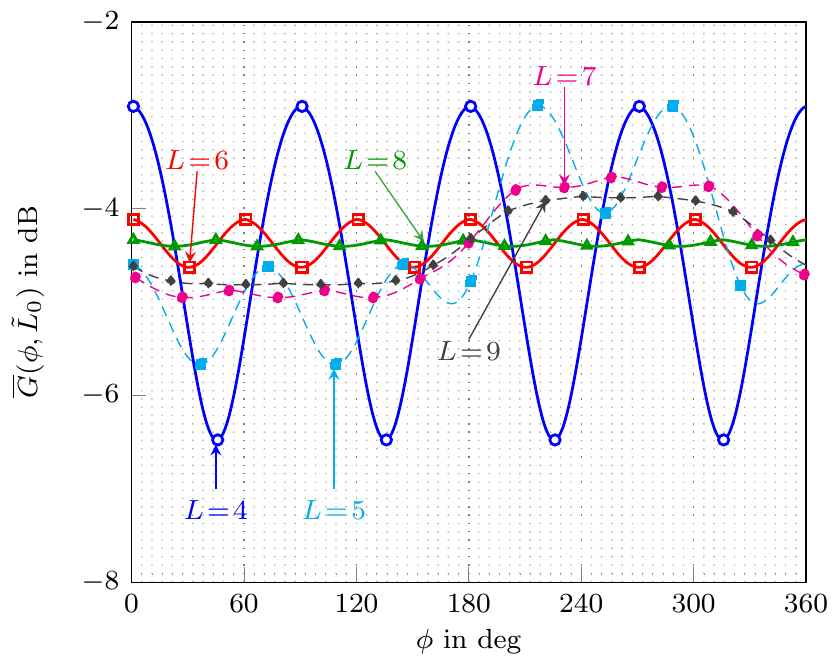}
	\caption{Effective radiation pattern for hybrid combiner with optimal configuration $(\ceil{L/2},\floor{L/2})$,  $L$ patch antennas and $K=10$.}	
	\label{omni}
\end{figure}

\subsection{Performance Bound}\label{VC}
	Given that the HC is used to combine the signals of antennas that have the same average gain in the azimuth plane $G_{\mathrm{2}\pi}$, Lemma~\ref{lem:limits} indicates that the performance of the system is upper bounded by $PG_{\mathrm{2}\pi}$. In the following, we would like to observe the implication of this lemma. 
	We use $L$ elements of the patch antenna deployed in Section~\ref{S5B} connected to the HC with $P=2$. The equivalent radiation patterns of the configurations $(\ceil{L/2},\floor{L/2})$ for $L=4,5,\ldots, 9$ are shown in \fig~\ref{omni}. These configurations are optimal as condition~\eqref{condition} is satisfied and $\ceil{L/2}\leq K=10$ for all of them. 
	From \fig~\ref{omni} we can notice that all the configurations with even $L$ are fluctuating around $2G_{\mathrm{2}\pi}=-4.37$\,dB. However,  increasing the number of antennas reduces the ripple (difference between maximum and minimum gain) of the equivalent radiation pattern, and yields better performance. We can observe that for $L=4$ there is a large ripple of approximately $3.6$\,dB. But starting from $L=6$ the ripple is significantly lower ($\approx 0.5$\,dB). Most importantly, the performance of the system with $L=8$ approaches very closely the performance upper limit of $2G_{\mathrm{2}\pi}=-4.37$\,dB and achieves good omnidirectional properties with $R(\Lzt)=0.07$\,dB. 
	Generally, we observe that performance is increased by adding an even number of antennas to the system. However,  the configurations with $L=7$ and $9$ have worse performance than the configuration with $L=6$. That is, performance can actually decrease by adding an odd number of antennas to a system with an even $L$. To be more precise, numerical results indicate that if $L$ is even and greater or equal to a certain threshold $L_{\mathrm{th}}$, then adding a single antenna (which results in an odd $L$) decreases performance. In \fig~\ref{omni}, the threshold is $L_{\mathrm{th}}=6$.%}
	
	To get more insight into the result of Lemma~\ref{lem:limits}, we plot the gain of the worst-case AOA of the equivalent radiation pattern with $L_0=\Lzt$ for different number of ADP2 antenna elements $L$ and different $\fiN$. We plot also the performance limit found in Lemma~\ref{lem:limits}. The results are shown in \fig~\ref{ADP2_PhiN_L}.  We can observe that performance saturates at the level $2G_{\mathrm{2}\pi}$ when the number of antennas is above certain threshold $L_{\mathrm{th}} $. We can identify that $L_{\mathrm{th}} = 6,~8$ and $12$, for the antennas with $\fiN=75$\,$\deg$, $50$\,$\deg$ and $35$\,$\deg$, respectively. We observe that antennas with wide main lobe (large $\fiN$) approaches the bound faster than antennas with narrow main lobe. 
	From another aspect, we can observe that configurations with even number of antennas faster approach the maximum possible gain of $2G_{\mathrm{2}\pi}$ at the worst-case AOA compared to configurations with odd number of antennas. 
	
	The results shown suggests that we need a relative small, even number of antennas to nearly reach the performance limit. Hence, the condition on the feasibility of HC that is $\ceil{L/2}\leq K$ can be easily met by the corresponding optimal adjacent scheme that has $\Lzt=\ceil{L/2}$. Also, another important fact is that our choice of an adjacent configuration in Section~\ref{S4} is not suboptimal, as the maximum performance can be achieved using the optimal adjacent configuration found in Theorem~\ref{th2}. Based on the results shown in \fig~\ref{ADP2_PhiN_L}, we can see that Lemma~\ref{lem:limits} gives the answers to the questions: Given directional antenna elements with particular $\fiB$ (or $\fiN$), how many antenna elements are needed for good performance? On the other hand, given $L$ antenna elements, how large should the main lobe of the antennas be, to reach good performance?

	\subsection{On Sidelobes Contribution to the Overall Performance } \label{SLvsML}
	We would like to assess the contribution of sidelobes to the overall system performance. To that end, let $G_{\mathrm{ML}}(\phi)$ be the main lobe of a directional antenna element,
	\begin{align}
	G_{\mathrm{ML}}(\phi)=
	\begin{cases}
	G(\phi),& \phi \in (-\fiB,\fiB),\\
	0,            & \phi \in [-\pi,-\fiB] \cup[\fiB,\pi ),
	\end{cases}
	\end{align} 
	and let $ \overline{ G}_{\mathrm{ML}}(\phi,L_0)$ be the equivalent radiation pattern computed based only on the main lobe contribution of antennas. That is, we substitute $G(\phi)$ by $G_{\mathrm{ML}}(\phi)$ in~\eqref{Gl} and~\eqref{G_eq}. Then, we compute $\inf_\phi \overline{ G}(\phi,L_0)$ and $\inf_\phi \overline{ G}_{\mathrm{ML}}(\phi,L_0)$ for $L=4$ and $L=10$ to span cases with small and large number of sidelobes. We compute the system performance  using ADP1 antennas that have relative large SLL and ADP3 antennas that are characterized by very small SLL. The results are shown in \fig~\ref{MLvsSL}. We notice that for ADP1 antennas with performance higher than $G_{\mathrm{SL}}$, the contribution of sidelobes is below $1$\,dB for the system with $L=10$ and the system with $L=4$ when $\fiN>60$\,$\deg$. While for all ADP3 the contribution of the sidelobes is less than $0.01$\,dB, which is negligible, for both $L=4$ and $L=10$. Hence we can conclude that the performance is mostly defined by the contribution of the main lobes. This back up the assumption put forward in the introduction, that directional antennas can be less affected by the car body effects if the antennas are mounted such that the distortion created by the car body and the antenna placement is mostly at the sidelobes regions.
		\begin{figure}[t]
		\center
		\includegraphics[width=\columnwidth]{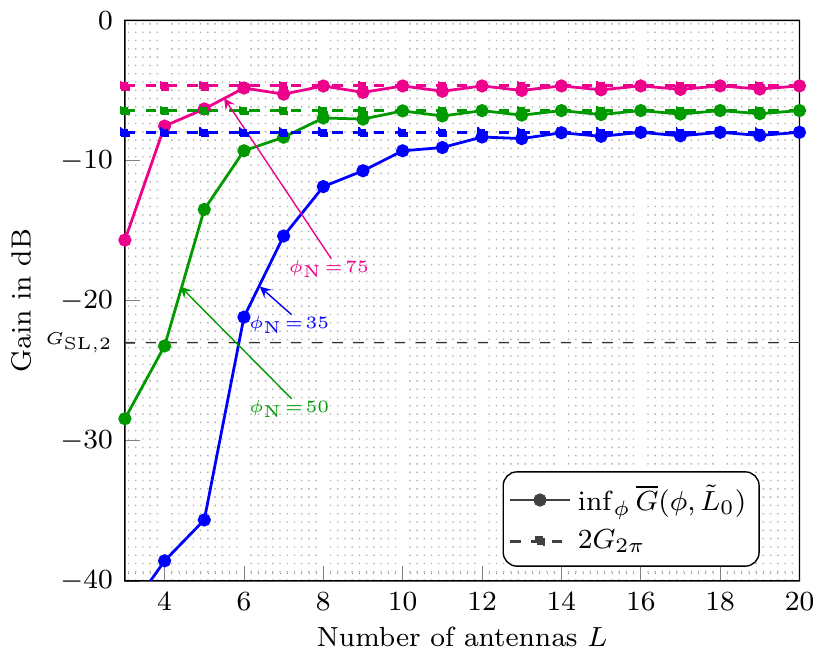}
		\caption{System performance as a function of the number of antenna elements $L$, computed for ADP2 patterns with $\fiN=35,50$, and $75 $\,$\deg$. The curves show the rate of approaching the performance bound.}	
		\label{ADP2_PhiN_L}
	\end{figure}
	\begin{figure}[t]
		\center
		\includegraphics[width=\columnwidth]{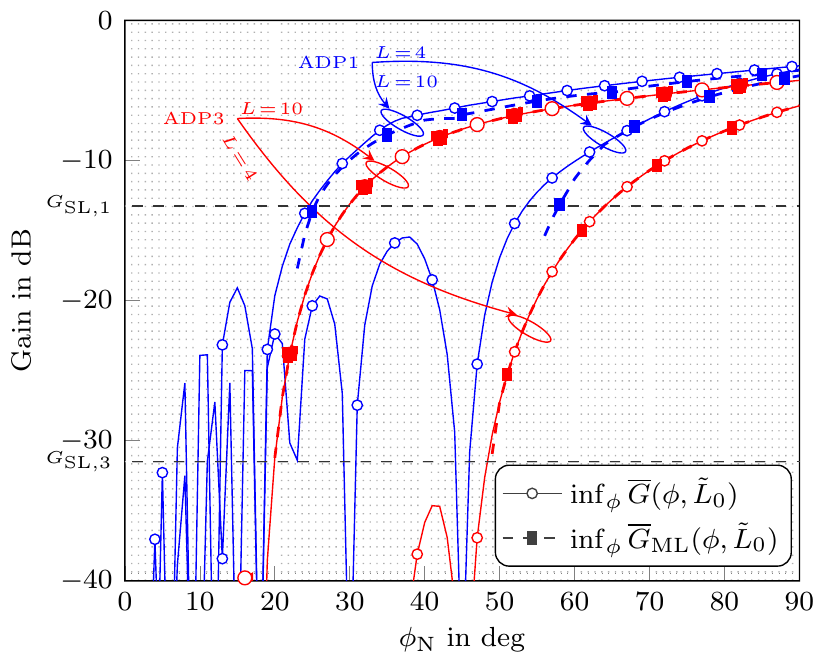}
		\caption{Comparison of the system performance when using full patterns and patterns without sidelobes. $L=4$ and $L=10$ elements of ADP1 and ADP3 patterns are used.}	
		\label{MLvsSL}
	\end{figure}		
\subsection{Comparison between Hybrid, Fully Analog and Fully Digital Combiners}
In this subsection, we would like to compare the performance of the HC with the performance of a fully digital and a fully analog combiner. We choose MRC with $P=L$ ports for the fully digital solution and the AC based on~\cite{mainACN} for the fully analog solution. We recall that the solution in~\cite{mainACN} assumes periodic broadcast V2V communication. It uses  a single port receiver and it combines the signals in absence of CSI using an analog phase shift network with phase slopes $\tilde{\boldsymbol{\alpha}}$. To derive the equivalent radiation pattern corresponding to these solutions, we first derive the SNR per packet. We use the same setting and assumptions as presented in Section~\ref{S2} and Section~\ref{S3}. Thus, the received signal at the antennas is given by~\eqref{SM:y} and the channel is modeled according to~\eqref{channel}. Then, for full MRC combining we have, 
\begin{align}
	\gamma_{\mathrm{D},k}=\frac{P_r}{\sigma_{\mathrm{n}}^2} \sum_{l=0}^{L-1}G_l(\phi),~k=0, 1, \ldots K-1,
\end{align}
where $ K$ corresponds to the number of packets in a burst and $P_\mathrm{r}=\mathbb{E}\{|a(t)x(t)|^2\}$.
For the fully analog combining of $L\leq K$ antennas using phase slopes $\tilde{\boldsymbol{\alpha}}$ and a single port, the SNR/packet can be derived in similar steps, and it is~\cite{mainACN}
\begin{align}
\gamma_{\mathrm{A},k}=\frac{P_r}{\sigma_{\mathrm{n}}^2} \sum_{l=0}^{L-1}G_l(\phi)/L,~k=0, 1,\ldots K-1.
\end{align}
After that, assuming an exponential PEP function as given in~\eqref{burst}, the the burst error probability of $K$ consecutive packets is given by~\eqref{PB}. Hence, the BrEP is inversely proportional to what we label the equivalent radiation patterns given by 
\begin{align}\label{G_D}
\overline{ G}_{\mathrm{D}}(\phi)=\sum_{l=0}^{L-1} G_l(\phi),\quad	\overline{ G}_{\mathrm{A}}(\phi)=\frac{1}{L}\sum_{l=0}^{L-1} G_l(\phi).
\end{align}
For the hybrid combining with $P=2$, we assume that $L$ even, and~\eqref{condition} holds. Then, the optimal configuration $(\Lzt,\tilde{L}_1)=(L/2,L/2)$ has equivalent radiation pattern $\overline{ G}(\phi,L/2)= 2/L \sum_{l=0}^{L-1} G_l(\phi) $. 

We can readily compare the different solutions. The gain of equivalent radiation pattern of the HC with $P=2$ is $3$\,dB higher, for all AOAs, than the gain of AC equivalent pattern, $\overline{ G}_{\mathrm{A}}(\phi)$. On the other hand, the gain of HC is $10\log_{10}(L/2)$\,dB lower, for all AOAs, than the gain of the full MRC equivalent pattern, $\overline{ G}_{\mathrm{D}}(\phi)$. Consequently, the equivalent patterns have the same omnidirectional characteristics but they differ in the gain. This can be observed in \fig~\ref{comp}. The results are based on deploying $L=4$ patch antennas that are spread around the azimuth plan. The normalized pattern of the patch antenna used is shown in \fig~\ref{pattern_dB}. The complexity of HC is in between the complexity of the other two counterparts. Hence, HC offers an attractive tradeoff between complexity and performance.
\begin{figure}[t]
	\center
	\includegraphics[width=\columnwidth]{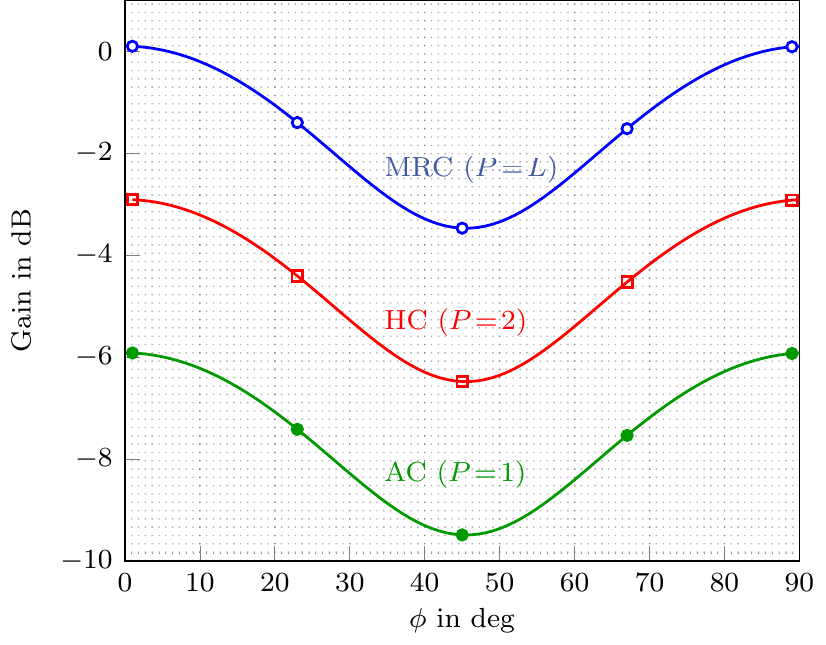}
	\caption{Comparison between the resultant equivalent pattern of the hybrid (HC), fully analog (AC)~\cite{mainACN} and fully digital (MRC) combiners, $L=4$.}	
	\label{comp}
\end{figure}

%%%%%%%%%%%%%%%%%%			%%%%%%%%%%%%%%%		%%%%%%%%%%%%%%%%%%%%%
\section {Conclusion}
We propose a hybrid combining scheme for directional antennas to achieve robust omnidirectional coverage for broadcast vehicle-to-vehicle communication. This has the promise to reduce the effects of both the car body and antenna placement on the omnidirectional characteristic of coverage. 
The digital part of the hybrid combiner (HC) is based on MRC, and its analog part is realized using phase shifters with slopes $\mathbf{\tilde{\alpha}}$ that achieve an upper bound on the burst error probability of consecutive status messages for the worst-case AOA. Given the use of $L$ directional antenna elements and the proposed hybrid combiner with two digital ports we conclude the following.
\begin{itemize}
	\item Under the fulfillment of the condition presented in Theorem~\ref{th2} on the side lobe level (SLL) of the antennas, the optimal way to group the antennas among the two ports is $(\ceil{L/2},\floor{L/2})$. This maximizes our performance metric, i.e., the gain of the equivalent radiation pattern at the worst-case AOA (which is equivalent to minimizing the burst error probability).
	\item Computation results of the condition presented in Theorem~\ref{th2} using three examples of analytical directional patterns with varying SLLs, indicates that it can be easily satisfied for antennas with reasonable SLLs. 
	\item Only a relative small number of antennas are needed to approach optimum HC performance (as defined in Lemma~\ref{lem:limits}). 
	\item Performance improves by increasing of $L$ with an even number of antennas. However, numerical results suggest that when $L$ is even and $L\ge L_{\mathrm{th}}$, then adding a single antenna to the system will decrease performance.
    \item To best approach the performance bound and for best omnidirectional characteristics of the equivalent radiation pattern, the number of directional elements per port should be the same, and thus $L$ should be \textit{even}.%, when using a two-port receiver. 
	\item The HC uses $(L-2)$ less ports than full MRC combining to achieve an equivalent radiation pattern with a similar omnidirectional properties when $L$ is even. However, the HC gain is $10\log_{10}(L/2)$\,dB less than full MRC (for any AOA). 
\end{itemize}

\begin{appendices}
	\section{Proof of Lemma~\ref{lem:limits} }
	\label{appendix_Lem1}
	\begin{proof}
		Let $\overline{ G}(\phi,\mathbf{S})$ defined as in~\eqref{G}, and let the average radiation of all antennas in the azimuth plane be the same,
		\begin{align}
		\frac{1}{2\pi}	\int_{2\pi} |g_l(\phi)|^2d\phi = G_{2\pi}, \quad\forall l.
		\end{align} 
		Then, it is easy to check that
		\begin{equation}
		\frac{1}{2\pi}	\int_{2\pi} \overline{ G}(\phi,\mathbf{S}) d\phi= PG_{2\pi}, \quad \mathbf{S}\in \mathcal{S}.
		\end{equation}
		By the monotonicity of Riemann integral it follows that 
		\begin{align}
		\int_{2\pi} \inf_{\phi\in[0, 2\pi)}\overline{G}(\phi,\mathbf{S})  d\phi\leq 
		\int_{2\pi}\overline{G}(\phi,\mathbf{S}) d\phi , \quad \mathbf{S}\in \mathcal{S}.
		\end{align}
		Thus, we can deduce that 
		\begin{align}
		\inf_{\phi\in[0, 2\pi)}
		\overline{G}(\phi,\mathbf{S}) \leq PG_{2\pi}, \quad \mathbf{S}\in \mathcal{S}.
		\end{align}
		By definition~\eqref{def:obj:Stilde}, $\tilde{\mathbf{S}} \in \mathcal{S}$. Replacing $\mathbf{S}$ by $\tilde{\mathbf{S}}$ in the previous equation,~\eqref{lem:lim1} follows.
	\end{proof}

\section{Proof of Theorem~\ref{th2}}
\label{appendix_Th2}
\begin{proof}
    We start by fixing $L$ and defining the function 
	\begin{align}\label{fa}
	    f_N(\phi)\triangleq\sum_{l=N}^{L-1}G_l(\phi),\quad  N=0,1, \ldots, L-1,
	\end{align}	
	where $G_l(\phi)=G(\phi-l2\pi/L)$ and $G(\phi)$ is as defined as in Theorem~\ref{th2}. We specifically recall that $\fiB<\pi/2$. 
    
    We can divide $G_l(\phi)$ into main lobe regions and sidelobes regions. The principal main lobe region for $G_l(\phi)$ is defined as  
	\begin{align}
	\mathcal{L}_l \triangleq\big(l2\pi/L-\fiB,l2\pi/L+\fiB \big), \quad l\in\mathbb{Z}. \label{eq:ML:def}
	\end{align}
	Clearly, $G_l(\phi) > G_{\mathrm{SL}}$ for $ \phi \in \mathcal{L}_l$. 
	In fact, since $G_l(\phi)$ is periodic with period $2\pi$,  $G_l(\phi)> G_{\mathrm{SL}}$ for $ \phi \in \mathcal{L}_{l+nL}$ for any integer $n$. 
	Since adjacent antennas are spaced with $2\pi/L$, their principal main lobe regions will overlap when $\fiB > \pi/L$, i.e., 
	\begin{align}
	    (\fiB > \pi/L)
	    &\Leftrightarrow (\inf\mathcal{L}_{l}  < \sup\mathcal{L}_{l-1}).
	    \label{eq:ML:overlap}
	\end{align}

	We note that $f_N(\phi)$ is periodic with period $2\pi$, and we will study its properties in the following interval of length $2\pi$,
	\begin{align}\label{big_I}
	\mathcal{I} &\triangleq [\sup\mathcal{L}_{-1},\sup \mathcal{L}_{L-1}).
	\end{align} 
	We define
	\begin{align}
	    \label{Ic}
    	\IzC(N) &\triangleq(\inf \mathcal{L}_N,\sup \mathcal{L}_{L-1}), &N=0,1, \ldots, L-1,
    	\end{align}
    	and 
    	\begin{align}
    	\Iz(N) &\triangleq \mathcal{I}\setminus \IzC(N), &N=0,1, \ldots, L-1 \label{I1},\\ 
    	&=[\sup \mathcal{L}_{-1}, \inf \mathcal{L}_N],&N=0,1, \ldots, L-1.
	\end{align}
	Hence, the principal main lobes of the $G_l(\phi)$ terms constituting $f_N(\phi)$ are contained in $\IzC(N)$, that is 
	\begin{align}\label{LinI0c}
	\mathcal{L}_l \subseteq	\IzC(N),\quad l=N, \ldots, L-1.
	\end{align}
	Conversely, the $G_l(\phi)$ terms in $f_N(\phi)$ have only sidelobes contributions in $\Iz(N)$. The combinations of $\fiB$ and $N$ that imply that $\Iz(N)\neq\emptyset$ are characterized by the following lemma. 
    \begin{lemma}\label{lem:IzN:nonempty}
        Let
    	\begin{align} 
        	\Mmin &\triangleq L\fiB/\pi-1, \label{eq:Mmin:def}\\
        	\M &\triangleq\{ m \in \mathbb{N}:\ceil{\Mmin}\leq m\leq L-1\} \label{M},
	    \end{align}
	    where $\mathbb{N} = \{0, 1, 2, \ldots\}$. 
	    Then the following conditions are equivalent
        \begin{equation}
            (\Iz(N)\neq\emptyset) \Leftrightarrow (N \in\M).
        \end{equation}
    \end{lemma}
    \begin{proof}
	    We see from the definitions \eqref{big_I}, \eqref{Ic}, and \eqref{I1} that $\Iz(N)$ is nonempty if, and only if, 
        $\inf \mathcal{L}_N \ge \sup \mathcal{L}_{-1}$, i.e., if $(N(2\pi/L)-\fiB \ge -(2\pi/L)+\fiB)$. The latter condition can be rewritten as $N\ge\Mmin$ or, since $N\le L-1$, as $N\in\M$.
    \end{proof}
    From the above lemma and~\eqref{LinI0c}, we conclude that the terms in $f_N(\phi)$ can be bounded as
	\begin{align} \label{I0sl} 
	G_l(\phi)\leq G_{\mathrm{SL}},\quad \phi\in\Iz(N), N\le l \le L-1, N\in \mathcal{M}. 
	\end{align}
     We can now restate the main condition for the theorem~\eqref{condition} by noting that $\mathcal{J}=\Iz(L-1)$, as
    \begin{align} \label{condition:restated}
    G_{\mathrm{SL}}   < \frac{\floor{L/2}}{\ceil{L/2}}
    \inf_{\phi \in \Iz(L-1)}	\frac{1}{L-1}\sum_{l=0}^{L-2}G_l(\phi).
	\end{align}
\begin{lemma}\label{lem:condition:fiB}
	If~\eqref{condition:restated} holds then $\fiB> \frac{\pi}{L}$.
	\end{lemma}
\begin{proof}
	We will use a proof of contradiction, i.e., we assume that \eqref{condition:restated} holds and $\fiB\le\pi/L$ and derive a contradiction.
	
	If $\fiB\leq \pi/L$,  then $\Mmin\leq 0$ and thus $\{0\} \in \mathcal{M}$, where  $\Mmin$ and $\mathcal{M}$ are defined according to~\eqref{eq:Mmin:def} and~\eqref{M}, respectively.
	By Lemma~\ref{lem:IzN:nonempty}, $\Iz(0)\neq \emptyset$, then~\eqref{I0sl} implies that
	\begin{align}
	G_{\mathrm{SL}} &\geq \inf_{\phi \in \Iz(0)}\frac{1}{L-1}\sum_{l=0}^{L-2}G_l(\phi)\\
    &\ge \inf_{\phi \in\Iz(L-1)}	\frac{1}{L-1}\sum_{l=0}^{L-2}G_l(\phi), \label{eq:GSL:greater:than}
	\end{align}
	where the last inequality holds since $\Iz(0) \subset \Iz(L-1)$.

	If~\eqref{condition:restated} holds and since $\floor{L/2}/\ceil{L/2} \le 1$ then 
	\begin{align}
	    G_{\mathrm{SL}} 
	    < \inf_{\phi\in\Iz(L-1)}\frac{1}{L-1} \sum_{l=0}^{L-2}G_l(\phi),
	\end{align}
	which contradicts~\eqref{eq:GSL:greater:than} and the lemma follows. 
\end{proof}
	
    The above lemma proves claim~(i) of the theorem. Following this result, we can assume $\fiB>\pi/L$ for the remainder of the proof.
	
   We will find it convenient to divide  $\IzC(N)$ to three parts. The middle part $\Io(N)$ is defined as
	\begin{align} \label{I2}
	\Io(N)\triangleq
	\begin{cases}
	[\sup \mathcal{L}_{N-1}, \inf \mathcal{L}_L], & \LN \geq \Mmin,\\
	[\inf \mathcal{L}_L,\sup \mathcal{L}_{N-1}], & \LN < \Mmin,\\
	\end{cases}
	\end{align}
    where
    \begin{equation}
        \LN\triangleq L - N \label{eq:Nbar:def}.
    \end{equation}
    The left-hand and right-hand parts of $\IzC(N)$ are defined as
	\begin{align}
	\Izo(N) &\triangleq \big(\inf \IzC(N), \inf \Io(N) \big), \label{I12}\\
	    &= \label{I12cases}
	    \begin{cases}
	        (\inf \mathcal{L}_N,\sup \mathcal{L}_{N-1}), & \LN \geq \Mmin,\\
	        (\inf \mathcal{L}_N,\inf \mathcal{L}_L), & \LN <\Mmin.\\
	    \end{cases}\\
	\Ioz(N) &\triangleq \big(\sup\Io(N) , \sup \IzC(N) \big ),\label{I21}\\
	    &=\label{I21cases}
	        \begin{cases}
	            (\inf \mathcal{L}_L, \sup \mathcal{L}_{L-1}), &\LN \geq \Mmin,\\
	            (\sup \mathcal{L}_{N-1},\sup \mathcal{L}_{L-1}), & \LN <\Mmin.\\
	    \end{cases}
	\end{align}
    The intervals $\Izo(N)$ and $\Ioz(N)$ are well-defined if $\Io(N)$ is nonempty, which indeed is the case as shown the following lemma. 
    \begin{lemma}
        \label{lem:IoN:nonempty}
        The following conditions are equivalent   
        \begin{equation}
            \label{eq:IzChatm:cond}
            (\inf\mathcal{L}_L \ge \sup\mathcal{L}_{N-1}) \Leftrightarrow (\LN \ge \Mmin).
        \end{equation}
        Moreover, $\Io(N)$ is nonempty.
    \end{lemma}
    \begin{proof}
        We have that
        \begin{align}
            \inf\mathcal{L}_L - \sup\mathcal{L}_{N-1} 
            &= L \frac{2\pi}{L}-\fiB - \left[(N-1)\frac{2\pi}{L}+\fiB\right]\nonumber\\
            &= \frac{2\pi}{L}\left[L-N - \left(\frac{L\fiB}{\pi}-1\right)\right]\nonumber\\
            &= \frac{2\pi}{L}[\LN - \Mmin], \label{eq:Ioz:length}
        \end{align}
        which proves~\eqref{eq:IzChatm:cond}. It follows from~\eqref{eq:Ioz:length} and the definition~\eqref{I2} that 
        $|\Io(N)| = |2\pi(\LN - \Mmin)/L| > 0$  when $\LN \neq \Mmin$. Moreover, $\Io(N)=\{2\pi-\fiB\}$ when $\LN = \Mmin$. Hence, $\Io(N)$ is nonempty.
    \end{proof}

    It is easily seen that,
    \begin{equation}
        \label{eq:IzChat:def}
        \IzC(N) = \Izo(N)\cup\Io(N)\cup\Ioz(N)
    \end{equation}
     since $\Izo(N)$ and $\Ioz(N)$ are nonempty when $\fiB > \pi/L$, as proven in the following lemma.

    \begin{lemma}
        \label{lem:Ioz:Ioz:emptiness}
        \begin{align}
            \fiB > \pi/L &\Rightarrow \text{$\Izo\neq\emptyset$ and $\Ioz\neq\emptyset$}.
        \end{align}
    \end{lemma}
    \begin{proof}

        Let $\fiB>\pi/L$, which according to~\eqref{eq:ML:overlap} implies that $\inf\mathcal{L}_{l}  < \sup\mathcal{L}_{l-1}$. 
        
        Suppose $\LN\ge\Mmin$. Since main lobe regions overlap, it follows from \eqref{I12cases} and \eqref{I21cases} that $\Izo\neq\emptyset$ and $\Ioz\neq\emptyset$.

        Suppose $\LN<\Mmin$. Since $N<L$, it follows immediately from \eqref{I12cases} and \eqref{I21cases} and the definition of $\mathcal{L}_l$ that $\Izo\neq\emptyset$ and $\Ioz\neq\emptyset$.
    \end{proof}
    We are now ready to state and prove a number of properties for $f_N(\phi)$.

\begin{figure*}[ht]
\centering
\subfigure[Case $\bar{N}\geq \Mmin$, ($N=4$)]
{
\includegraphics[width=0.9\columnwidth]{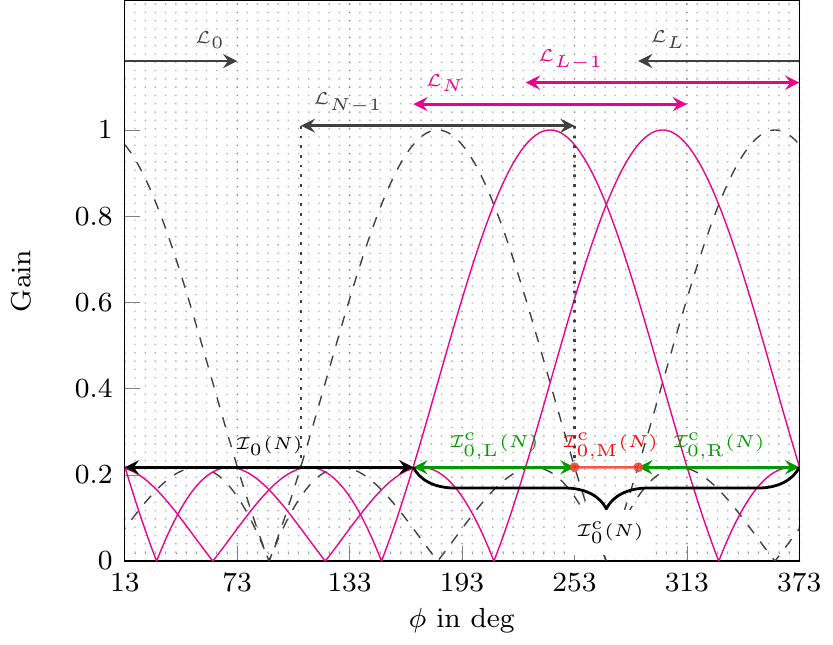}
}
\subfigure[Case $\bar{N}<\Mmin$, ($N=5$)]
{
\includegraphics[width=0.9\columnwidth]{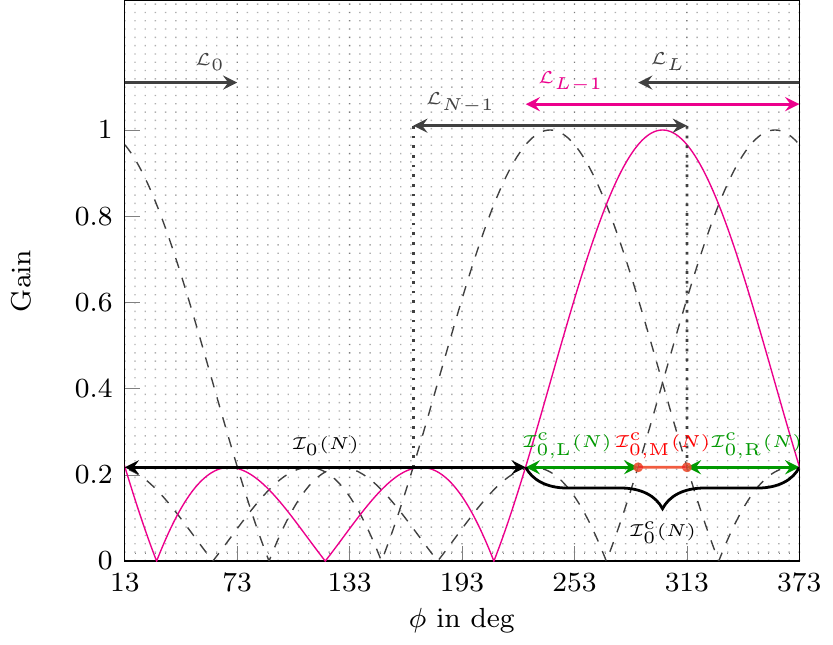}
}
\caption{Illustration of the intervals $\Iz(N)$, $\IzC(N)$, $\Io(N)$, $\Ioz(N)$ and $\Izo(N)$ for different cases, $L=6$, $\fiB=73$\,$\deg$.}	
\label{}
\end{figure*}

	\begin{lemma}\label{Lem_f}
		Fix $L$. Let $f_N(\phi)$ be defined as in~\eqref{fa}, where $N=0, \ldots, L-1$, and let
		
		$\LN=L-N$. Then, $f_N$ has the following properties.
		\begin{enumerate}[label=(\roman*)]
			\item\label{Lem_f:p1} For $n= 0, 1, \ldots, \LN-1$
			\begin{align}\label{Pi0}
			f_N\big(\phi-n\frac{2\pi}{L}\big)=\sum_{l=N+n}^{L-1}G_l(\phi)+\sum_{l=0}^{n-1}G_l(\phi),
			\end{align}
			and 
			\begin{align}\label{Pi1}
			f_N\big(\phi-\LN \frac{2\pi}{L}\big)&=\sum_{l=0}^{\LN-1}G_l(\phi)=f_N\big(\phi+N\frac{2\pi}{L}\big).
			\end{align}
			\item\label{Lem_f:p2}   $f_0(\phi)=\sum_{l=0}^{L-1}G_l(\phi)$ is periodic with period $2\pi/L$,
			\begin{align}
			f_0(\phi-n2\pi/L)=f_0(\phi), \quad n\in \mathbb{Z}.
			\end{align}
			
			\item\label{Lem_f:p3} If $N \in \mathcal{M}$, then
		
			\begin{align}\label{Piii}
			f_N(\phi)\leq (L-N)G_{SL} ,\quad  \phi \in \Iz(N),
			\end{align}
			where $\Iz(N)$ is defined in~\eqref{I1}.
		
			\item\label{Lem_f:p4}
			If~\eqref{condition:restated} holds and $Q \in \mathcal{M}$, then 
			\begin{align}
			f_0(\phi) -f_Q(\phi) &= \sum_{l=0}^{Q -1}G_l(\phi) \nonumber\\
			&> \frac{\ceil{L/2}}{\floor{L/2}}Q G_{\mathrm{SL}},\quad \phi \in \Iz(Q).\label{con_f}
			\end{align}
			\item \label{Lem_f:p5}  If~\eqref{condition:restated} holds and $N\neq 0$, then 
			\begin{align} \label{cond1}
			f_{N}(\phi)>\LN G_{\mathrm{SL}}, \quad \phi \in  \Io(N),
			\end{align}
			where $\Io(N)$ is defined in~\eqref{I2}.
			\item\label{Lem_f:p6} For $N\geq \ceil{L/2}$, 
			\begin{align}
			f_N(\phi)>f_N\big(\phi-2\pi/L\big), ~	\phi \in \Izo(N),\label{cond02}
			\end{align}
			where $\Izo(N)$ is defined according to~\eqref{I12}.
			\item\label{Lem_f:p7} For $N\geq \ceil{L/2}$, 
			\begin{align}
			f_N(\phi)>f_N(\phi+2\pi/L),~	\phi \in \Ioz(N), \label{cond03}
			\end{align}
			where $\Ioz(N)$ is defined according to~\eqref{I21}.
		\end{enumerate}

	\end{lemma}
	\begin{proof}
	    Property \ref{Lem_f:p1}:
		Let $n=0, 1, \ldots, \LN-1$. By definition of $f_N(\phi)$, we have
			\begin{align}
			f_N\big(\phi-n\frac{2\pi}{L}\big)&=\sum_{l=N}^{L-1}G(\phi-(l+n)2\pi/L)\nonumber\\
			&=\sum_{i=N+n}^{L-1}G_{i}(\phi)+\sum_{i=L}^{L+n-1}G_{i}(\phi).\label{shift1}
			\end{align}
			Since $G(\phi)$ is periodic, the second sum is equal to 
			\begin{equation}
			\sum_{l=0}^{n-1}G(\phi-(l+L)\frac{2\pi}{L})=\sum_{l=0}^{n-1}G_l(\phi),
			\end{equation}
			which proves~\eqref{Pi0}. To prove~\eqref{Pi1},  we have
			\begin{align}
			f_N\big(\phi-\LN\frac{2\pi}{L}\big)&=\sum_{l=N}^{L-1}G(\phi-(l+L-N)2\pi/L)\nonumber\\
			&=\sum_{l=N}^{L-1}G(\phi+N2\pi/L - l2\pi/L)\label{Pi1:half}\\
			&=\sum_{l=0}^{\LN-1}G(\phi - l2\pi/L)\label{Pi1:half2}.
			\end{align}
			From another aspect,~\eqref{Pi1:half} can be expressed as
			\begin{align}
					\sum_{l=N}^{L-1}G(\phi+N2\pi/L - l2\pi/L)=f_N(\phi+N2\pi/L).
			\end{align}
			Putting this together with~\eqref{Pi1:half2} proves~\eqref{Pi1} and the property follows.

			Property \ref{Lem_f:p2}: Let $n=kL+ n'$, where $n', k \in \mathbb{Z}$ and $0\leq n'\leq (L-1)$. Then,
			\begin{align}
			f_0\big(\phi- n2\pi/L\big)&=f_0\big(\phi- n'2\pi/L\big)\label{seq:p2} \\
			&=\sum_{l=n'}^{L-1}G_l(\phi)+\sum_{l=0}^{n'-1}G_l(\phi)\label{ii_1}\\
			&=f_0(\phi), 
			\end{align}
			where~\eqref{seq:p2} holds since $f_0$ is periodic with period $2\pi$, and~\eqref{ii_1} follows from~\eqref{Pi0} (since $N=0$, $n'< \LN =L$).
			
			Property \ref{Lem_f:p3}:
			Follows directly from~\eqref{I0sl}.
            
            Property \ref{Lem_f:p4}: 
			By assumption~\eqref{condition:restated} holds. It implies that 
			\begin{align}
			\sum_{l=0}^{(L-1)-1}G_l(\phi)> a (L-1) G_{\mathrm{SL}}, \quad \phi \in \Iz(L-1), \label{Lem_f_p:p4:eq0}
			\end{align}
			where $a= \ceil{L/2}/\floor{L/2}\ge 1$. 
			Hence the property holds for $Q=L-1$. 
			Now, let $Q < L-1$, then for $\phi\in\Iz(L-1)$,
			\begin{align}
			\sum_{l=0}^{Q-1}G_l(\phi) &= \sum_{l=0}^{L-2} G_l(\phi) - \sum_{l=Q}^{L-2} G_l(\phi)\nonumber\\
			&> a (L-1) G_{\mathrm{SL}} - \sum_{l=Q}^{L-2} G_l(\phi) \label{eq:inequality}\\
			&= a Q G_{\mathrm{SL}} +a(L-1-Q)G_{\mathrm{SL}} 
			- \sum_{l=Q}^{L-2} G_l(\phi), \label{Lem_f_p:p4:eq1}
			\end{align}
			where~\eqref{eq:inequality} follows from~\eqref{Lem_f_p:p4:eq0}.
			Since $Q < L-1$, it follows from the definition of $\Iz(N)$ that $\Iz(Q) \subset \Iz(L-1)$. Hence, \eqref{Lem_f_p:p4:eq1} holds in particular for $\phi\in\Iz(Q)$. What remains is to show that the sum of last two terms in \eqref{Lem_f_p:p4:eq1} is nonnegative. To this end, let $\phi\in\Iz(Q)$. Then by~\eqref{I0sl}, we have that $G_l(\phi)\leq G_{\mathrm{SL}}$, for $l=Q, Q+1, \ldots, L-1$ and 
			\begin{align}
			a(L-1&-Q)G_{\mathrm{SL}} 
			- \sum_{l=Q}^{L-2} G_l(\phi) \nonumber\\
		    &\ge a(L-1-Q)G_{\mathrm{SL}} - (L-1-Q)G_{\mathrm{SL}} \nonumber\\
		    &= (a-1)(L-1-Q)G_{\mathrm{SL}} \nonumber\\
		    &\ge 0,
			\end{align}  
            since $a\ge 1$, and the property follows.
            
			Property \ref{Lem_f:p5}: Let $N\neq 0$, then $\LN\leq L-1$. Based on the definition of $\Io(N)$ we divide the proof into two cases $\LN\geq \Mmin$ and $\LN<\Mmin$.
			
			We start tackling the case $\LN\geq \Mmin$. From~\eqref{I2} we have $\Io(N)=[\sup \mathcal{L}_{N-1}, \inf \mathcal{L}_L ]$. Since~\eqref{condition:restated} holds and $\Mmin \leq \LN \leq L-1$, that is $\LN\in \mathcal{M}$, Lemma~\ref{Lem_f} property %\ref{Lem_f:p1} and
			\ref{Lem_f:p4} is applicable with $Q=\LN$. We can use~\eqref{Pi1},~\eqref{con_f}, and the fact that $\ceil{L/2}/\floor{L/2}\ge 1$ to show
			\begin{equation}
			    f_N(\phi' + N2\pi/L) > \LN G_{\mathrm{SL}}, \qquad \phi' \in \Iz(\LN) \label{Lem_f_p:p5:eq1}.
			\end{equation}
			Now if $\phi' \in \Iz(\LN)$, then by definition, 
			\begin{equation}
			    - 2\pi/L+\fiB  \le \phi' \le\LN 2\pi/L - \fiB.
			\end{equation} 
			Adding $N2\pi/L$ to the inequality sides we get,
			\begin{align}
				\phi' + N2\pi/L\ge 	(N-1)2\pi/L + \fiB \nonumber \\
				\phi' + N2\pi/L\le (\LN+N) 2\pi/L - \fiB,
			\end{align}
    		which, since $\LN + N = L$, is equivalent to
			\begin{align*}
			     \sup \mathcal{L}_{N-1} &\le \phi' + N2\pi/L \le \inf \mathcal{L}_L.
			 \end{align*}
			Thus, if $\phi' \in \Iz(\LN)$, then $\phi = (\phi' + N2\pi/L) \in \Io(N)$ and  
			\eqref{Lem_f_p:p5:eq1} can be written as
			\begin{equation}
			    f_N(\phi) > \LN G_{\mathrm{SL}}, \qquad \phi \in \Io(N),
			\end{equation}
			which proves the property for the case when $\LN\geq \Mmin$.

			We now move to the case $\LN<\Mmin$ where $\Io(N)=[\inf\mathcal{L}_L, \sup \mathcal{L}_{N-1}]$. From the definition of $\mathcal{L}_l$ it follows that for $l=N, N+1,\ldots L-1$,
			\begin{align}
			\inf \Io(N) &= \inf\mathcal{L}_L > \inf  \mathcal{L}_l\\
			\sup \Io(N) &= \sup \mathcal{L}_{N-1} < \sup \mathcal{L}_l,
			\end{align}
			which implies that $\Io(N) \subset \mathcal{L}_l$ for $l=N, N+1,\ldots L-1$ and
			\begin{equation}
			    f_N(\phi) = \sum_{l=N}^{L-1} G_l(\phi) > (L-N) G_{\mathrm{SL}}, \qquad \phi\in\Io(N). 
			\end{equation}
			Hence, \eqref{cond1} holds when $\LN<\Mmin$ as well, and the property follows.
			
			Property \ref{Lem_f:p6}: We recall that $\fiB> \pi/L$, which implies $\Ioz(N)\neq \emptyset$ (Lemma~\ref{lem:Ioz:Ioz:emptiness}). Using~\eqref{Pi0}, we have that
			\begin{align}
			f_N(\phi-2\pi/L)&=\sum_{l=N+1}^{L-1}G_l(\phi)+ G_0(\phi),
			\end{align}
			and
			\begin{align}\label{f-f}
			f_N(\phi)-f_N(\phi-2\pi/L)&= G_N(\phi)-G_0(\phi).
			\end{align}
			From here, we prove~\eqref{cond02} by showing that $G_N(\phi)$ is a main lobe term and $G_0(\phi)$ is a sidelobes term when $\phi \in \Izo(N)$.  We demonstrate this for the two cases $\LN \geq \Mmin$ and $\LN<\Mmin$. 
			
			Let $\LN \geq \Mmin$. By definition~\eqref{I12cases}, we have that $\Izo(N)=(\inf \mathcal{L}_N, \sup \mathcal{L}_{N-1})$. Since
			 $\sup \mathcal{L}_{N-1}<\sup \mathcal{L}_{N}$, we see that $ \Izo(N)\subset  \mathcal{L}_N$, and 
			\begin{align}
			G_N(\phi)&>G_{\mathrm{SL}},\quad  \phi \in \Izo(N).
			\end{align}
			Then it is enough to show that $G_0(\phi)\leq G_{\mathrm{SL}}$, when $\phi \in \Izo(N)$ for~\eqref{cond02} to hold. This is equivalent to $\Izo(N) \subset (\sup\mathcal{L}_0, \inf \mathcal{L}_L)$.
			Since $N\geq \ceil{L/2}$, and we recall that $\fiB<\pi/2$, then $\inf \Izo(N)=N2\pi/L-\fiB\geq \pi-\fiB>\fiB=\sup\mathcal{L}_0$. Hence, 
			\begin{align}
			\inf \Izo(N)>\sup\mathcal{L}_0.
			\end{align}
			On the other hand, since $\LN \geq \Mmin$, then by~\eqref{eq:IzChatm:cond}%~\eqref{Ll_Ln:1} 
			\begin{align}
			\sup \Izo(N)=\sup \mathcal{L}_{N-1}\le\inf \mathcal{L}_L,
			\end{align}
			and thus we can conclude that $G_0(\phi)\le G_{\mathrm{SL}}$ for $\phi\in\Izo(N)$, and the property is proved for $\LN \geq \Mmin $.
	
			Now we proceed to the case $\LN< \Mmin$. By definition we have that $\Izo(N)=(\inf \mathcal{L}_N,  \inf \mathcal{L}_L)$. Then, again, 
			since  $N\geq \ceil{L/2}$, we have that $\inf \Izo(N)>\sup\mathcal{L}_0$ and, therefore, $G_0(\phi)\leq G_{\mathrm{SL}}$ for $\phi \in \Izo(N)$.
			Moreover, by~\eqref{eq:IzChatm:cond}, we have $\inf \mathcal{L}_L<\sup \mathcal{L}_{N-1}<\sup \mathcal{L}_{N}$. Thus, it follows that $ \Izo(N)\subset  \mathcal{L}_N$ and  $G_N(\phi)>G_{\mathrm{SL}}$ for $\phi\in\Izo(N)$. Hence,~\eqref{cond02} holds when $\LN < \Mmin $ as well and this ends the proof of the property.	
			
			Property \ref{Lem_f:p7}: Following the same steps as used to prove~Property \ref{Lem_f:p6}, we can write $f_N(\phi)-f_N(\phi+2\pi/L)= G_{L-1}(\phi)-G_{N-1}(\phi)$. Then we can deduce that $\Ioz(N) \subset \mathcal{L}_{L-1}$ and $\Ioz(N) \subset (\sup\mathcal{L}_{N-1},\inf \mathcal{L}_{N-1}+2\pi)$ and therefore~\eqref{cond03} holds.
			
	\end{proof} 
	
	\begin{lemma} \label{Lem:L0inM}
	    Let 
	    \begin{equation}
	        \label{eq:Lplus:def}
	        \mathcal{L}^+ \triangleq\{l\in\mathbb{Z}: L/2< l\le  L-1\}.
	    \end{equation}
	    If $L_0\in\mathcal{L}^+$, then $L_0\in\M$.
	\end{lemma}
	\begin{proof}
		We recall that $\pi/L<\fiB<\pi/2$. Let $L_0\in\mathcal{L}^+$. We have $\min\M =\ceil{ \Mmin} = \ceil{L\fiB/\pi}-1 < \ceil{L/2}-1 < L/2 < L_0$.
       Hence, $\min\mathcal{M} < L_0 \le \max\mathcal{M}$, and the lemma follows. 
	\end{proof}

	\begin{lemma} \label{Lem_inf}
		Let $\overline{ G}(\phi,L_0)$ and $\Iz(L_0)$ be defined according to~\eqref{G_eq} and \eqref{I1}, respectively. If 
		$L_0\in\mathcal{L}^+$ and~\eqref{condition:restated} is satisfied, then 
		\begin{equation}\label{eqLem}
		\inf_{\phi\in [0,2\pi)}	\overline{ G}(\phi,L_0) = \inf_{\phi \in \Iz(L_0)} \overline{ G}(\phi,L_0).
		\end{equation}
	\end{lemma}

	\begin{proof}	
	    Since $\overline{ G}(\phi,L_0)$ is periodic and $|\mathcal{I}|=2\pi$, 
	    \begin{align}
		\phi_0
		=\arg \inf_{\phi \in [0, 2\pi) } \overline{ G}(\phi,L_0)
		=\arg \inf_{\phi \in \mathcal{I} } \overline{ G}(\phi,L_0),
		\end{align} 
		where $\mathcal{I}$ is defined in~\eqref{big_I} and $\mathcal{I}=\Iz(L_0) \cup \IzC(L_0)$, where $\IzC(L_0)$ is defined according to~\eqref{Ic}.
		Hence, to prove the lemma, it is sufficient to show that $\phi_0 \in \Iz(L_0)$ or, equivalently, that
		\begin{align}\label{lemclaim}
		\phi_0 \notin \IzC(L_0)
		\end{align}	 
		since $\IzC(L_0)=\mathcal{I}\setminus \Iz(L_0)$.

        We express $\overline G(\phi,L_0)$ in terms of $f_N(\phi)$, where these are defined according to~\eqref{G} and~\eqref{fa}, respectively, as	
		\begin{equation} \label{inf_for}
		\begin{split}
		\overline{ G}(\phi,L_0)&= \big(f_0(\phi)+ b f_{L_0}(\phi)\big)\big/L_0,
		\end{split}
		\end{equation} 
		with $b=(2L_0-L)/(L-L_0)>0$. 
		
         We note that Lemma~\ref{Lem_f} property~\ref{Lem_f:p4} and \ref{Lem_f:p5} holds in both cases. Also, since $L_0>L/2$ implies that $L_0\ge\ceil{L/2}$, Lemma~\ref{Lem_f} property \ref{Lem_f:p6} and \ref{Lem_f:p7} holds.
        Moreover, we know from Lemma~\ref{Lem:L0inM} that $L_0\in\M$. Lastly, 
         we recall that~\eqref{condition:restated} implies $\fiB>\pi/L$ (Lemma~\ref{lem:condition:fiB}), and hence we let $\fiB> \pi/L$ and proceed to show~\eqref{lemclaim}.
	
		By~\eqref{eq:IzChat:def} we have  
		$\IzC(L_0)=\Izo(L_0)\cup \Io(L_0)\cup \Ioz(L_0)$, where these intervals are nonempty and defined according to~\eqref{I12},~\eqref{I2} and~\eqref{I21}, respectively.
		To prove~\eqref{lemclaim}, 
		we use contradictions to show that $\phi_0 \notin \Io(L_0)$,  $\phi_0 \notin \Izo(L_0) $ and $\phi_0 \notin \Ioz(L_0)$. 
		
		\underline{First, suppose $\phi_0 \in \Io(L_0)$}, Lemma~\ref{Lem_f} property \ref{Lem_f:p5} implies
		\begin{align}
		f_{L_0}(\phi) &> (L-L_0)G_{\mathrm{SL}}, ~\phi \in \Io(L_0).
		\end{align}
		Substituting in~\eqref{inf_for} and taking into account the assumption that $\phi_0\in \Io(L_0)$, we get
		\begin{align}\label{G_fiz_I1}
		\overline{ G}(\phi_0,L_0)>\big(f_0(\phi_0)+ (2L_0-L)G_{\mathrm{SL}} \big)\big/L_0. 
		\end{align}
		Now we attempt to find a $\phi' \notin \Io(L_0)$ such that $\overline{ G}(\phi',L_0)<\overline{ G}(\phi_0,L_0)$. Since $L_0>L/2$ then
		\begin{align}\nonumber
			|\Iz(L_0)|=(L_0+1)2\pi/L-2\fiB>\pi+2\pi/L-2\fiB.
		\end{align} 
		 Recalling that $\fiB<\pi/2$, we get $|\Iz(L_0)|>2\pi/L$. Adding that to the fact that $f_0(\phi)$ is periodic with $2\pi/L$ (Lemma~\ref{Lem_f} property~\ref{Lem_f:p2}), we deduce that $\exists$ $\phi'  \in  \Iz(L_0) $ such that $f_0(\phi') = \inf f_0(\phi) $. On the other hand, using Lemma~\ref{Lem_f} property~\ref{Lem_f:p3} we have
		\begin{align}
		f_{L_0}(\phi) < (L-L_0) G_{\mathrm{SL}},  \quad\phi \in \Iz(L_0). 
		\end{align}
		Following this we can upper bound $\overline{ G}(\phi',L_0)$ as
		\begin{align}\label{G_fip_I1}
		\overline{ G}(\phi',L_0)	< \big(f_0(\phi')+ (2L_0-L)G_{\mathrm{SL}} \big) \big/L_0. 
		\end{align}
		Comparing~\eqref{G_fip_I1} with~\eqref{G_fiz_I1}  and taking into account the fact that $f_0(\phi')\leq f_0(\phi_0)$ we deduce that 
		\begin{align}
		\overline{ G}(\phi',L_0)< 	\overline{ G}(\phi_0,L_0) .
		\end{align}	
		This contradicts the claim that the infimum of $\overline{G}(\phi,L_0)$, $\phi_0 \in \Io(L_0) $, therefore $\phi_0 \notin \Io(L_0)$.
		
		\underline{Secondly, suppose $\phi_0 \in \Ioz(L_0)$},
		to find a contradiction to this claim, we let $\phi'=\phi_0-2\pi/L$. By periodicity of $f_0(\phi)$, we have that $f_0(\phi')=f_0(\phi_0)$, then
		\begin{align}
		\overline{ G}(\phi',L_0)&= \big(f_0(\phi_0)+ b f_{L_0}(\phi') \big) \big/L_0.
		\end{align}  
		Moreover, by Lemma~\ref{Lem_f} property~\ref{Lem_f:p6} we have $f_{L_0}(\phi')< f_{L_0}(\phi_0)$. Therefore,
		\begin{align}
		\overline{ G}(\phi',L_0)&<  \big(f_0(\phi_0)+ bf_{L_0}(\phi_0) \big) \big/L_0=\overline{ G}(\phi_0,L_0).
		\end{align}
		This is a contradiction. So we can conclude that $\phi_0 \notin \Izo(L_0)$. 
		
		\underline{Lastly, suppose $\phi_0 \in \Ioz(L_0)$},
		by similar argument to the previous one, if we let $\phi'=\phi_0+2\pi/L$ and use Lemma~\ref{Lem_f} property~\ref{Lem_f:p7} we can deduce that $\overline{ G}(\phi',L_0)<\overline{ G}(\phi_0,L_0)$ and thus $\phi_0 \notin \Ioz(L_0)$.  
		
		We have shown that if~\eqref{condition:restated} holds then $\phi_0 \notin \Io(N)$, $\phi_0\notin \Ioz(N)$ and $\phi_0 \notin \Izo(N)$. Therefore, $\phi_0\notin \IzC(L_0)$ and this concludes the proof of the Lemma.
	\end{proof}
Now we are set to show the main result of Theorem~\ref{th2}, i.e., claim (ii). We recall that claim (i) was already proven in Lemma~\ref{lem:condition:fiB}. By assumption, the gain of the highest sidelobes $G_{\mathrm{SL}}$ satisfies~\eqref{condition}, which is restated in~\eqref{condition:restated}.  Then, we want to prove that the solution to~\eqref{L1_opt}, with the constraint $L_0\leq K$ relaxed, is~\eqref{L1*}, that is $\Lzt=\lceil L/2 \rceil$, $\tilde{L}_1=L-\Lzt = \floor{L/2}$. To prove this, it is enough to show that for $L_0\in\mathcal{L}^{+}\setminus\{\ceil{L/2}\}$, 
\begin{equation}\label{infCond}
\inf_{\phi\in[0, 2\pi)}
\overline{ G}(\phi,L_0)  <  \inf_{\phi\in [0, 2\pi) } \overline{ G}(\phi,\ceil{L/2}),
\end{equation} 
where $\overline{ G}(\phi,L_0)$ is the equivalent radiation pattern given by~\eqref{G_eq} and $\mathcal{L}^+$  is defined in~\eqref{eq:Lplus:def}.  

We tackle the demonstration of~\eqref{infCond} in two steps. In the first step we will show that for $L_0\in\mathcal{L}^+\setminus\{\ceil{L/2}\}$,
\begin{equation}\label{halfway_0}
\inf_{\phi\in [0,2\pi)}	\overline{ G}(\phi,L_0) < \inf_{\phi \in \Iz(L_0)} \overline{ G}(\phi,\ceil{L/2}).
\end{equation}
And in the second step, we will demonstrate that
\begin{equation}\label{InfLodd}
\inf_{\phi\in [0,2\pi]}	\overline{ G}(\phi,\ceil{L/2}) = \inf_{\phi \in \Iz(L_0)} \overline{ G}(\phi,\ceil{L/2}).
\end{equation}
We start the first step by forming 
\begin{align}
D(L_0) &= \overline{ G}(\phi,\ceil{L/2})-\overline{ G}(\phi,L_0) \nonumber\\
&=\bigg(\frac{1}{\ceil{L/2}}-\frac{1}{L_0} \bigg ) \sum_{l=0}^{\ceil{L/2}-1}G_l(\phi) \nonumber \\
&\quad +\bigg(\frac{1}{\floor{L/2}} -  \frac{1}{L_0} \bigg)  \sum_{l=\ceil{L/2}}^{L_0-1}G_l(\phi)  \nonumber\\
&\quad+\bigg(\frac{1}{\floor{L/2}} - \frac{1}{L-L_0}\bigg) \sum_{l=L_0}^{L-1}G_l(\phi) \nonumber \\ 
&\geq \bigg(\frac{1}{\ceil{L/2}}-\frac{1}{L_0} \bigg) \sum_{l=0}^{L_0-1}G_l(\phi) \nonumber \\
&\quad+\bigg(\frac{1}{\floor{L/2}}-\frac{1}{L-L_0} \bigg) \sum_{l=L_0}^{L-1}G_l(\phi) \nonumber\\
&= b \bigg(\sum_{l=0}^{L_0-1}G_l(\phi)-a\frac{L_0}{L-L_0} f_{L_0}(\phi)\bigg), \label{diff2}
\end{align}		
where the inequality follows since $1/\floor{L/2}\ge1/\ceil{L/2}$ and where
$f_{L_0}(\phi)$ is defined in~\eqref{fa}, $a=\ceil{L/2}/\floor{L/2}$ and $b=(L_0-\ceil{L/2})/(L_0\ceil{L/2}) $.  

From Lemma~\ref{Lem:L0inM} we can deduce that $L_0\in\mathcal{L}^+\setminus\{\ceil{L/2}\}$, implies $L_0\in\mathcal{M}$. Hence, we can readily use Lemma~\ref{Lem_f} property~\ref{Lem_f:p3} to get
\begin{align}
    \label{eq:fL0:bound}
    f_{L_0}(\phi) \leq (L-L_0)G_{\mathrm{SL}},\quad  \phi \in  \Iz(L_0).
\end{align}
Moreover, since~\eqref{condition:restated} holds, Lemma~\ref{Lem_f} property~\ref{Lem_f:p4} is applicable with $Q=L_0$ and 
\begin{align}
        \label{eq:sum:Gl:bound}
    \sum_{l=0}^{L_0-1} G_l(\phi) > a L_0 G_{\mathrm{SL}},\quad \phi \in \Iz(L_0).
\end{align}
Substituting~\eqref{eq:fL0:bound} and~\eqref{eq:sum:Gl:bound} into~\eqref{diff2}, and noting that $b>0$ for $L_0\in\mathcal{L}^+\setminus\{\ceil{L/2}\}$, yields that $D(L_0)>0$ for $\phi\in\Iz(L_0)$. This implies that
\begin{equation}\label{halfway}
\inf_{\phi \in \Iz(L_0)} \overline{ G}(\phi,L_0)< \inf_{\phi \in \Iz(L_0)} \overline{ G}(\phi,\ceil{L/2}).
\end{equation}
Finally, since
\begin{equation}
    \inf_{\phi\in [0,2\pi)}	\overline{ G}(\phi,L_0) \leq \inf_{\phi \in \Iz(L_0)} \overline{ G}(\phi,L_0),
\end{equation}
we have shown that~\eqref{halfway_0} holds. 
By this we are halfway through the proof of~\eqref{infCond}. To complete it we proceed to the second step: to show~\eqref{InfLodd}. We divide the demonstration in two cases, the case $L$ is even and the case $L$ is odd. 

Suppose $L$ is even, then according to Lemma~\ref{Lem_f} property~\ref{Lem_f:p2} $\overline{ G}(\phi,L/2)=2f_0(\phi)/L$ is periodic with period $2\pi/L$. 
Moreover, since $\fiB<\pi/2$, then $| \Iz(L_0)|>2\pi/L$ and~\eqref{InfLodd} holds (when $L$ is even).

Suppose $L$ be odd, then $\min\mathcal{L}^+ = \ceil{L/2}$.  Lemma~\ref{Lem_inf}, for the special case $L_0=\ceil{L/2}$, yields
\begin{align}
    \label{eq:inf:LO-lower}
\inf_{\phi\in [0,2\pi]}	\overline{ G}(\phi,\ceil{L/2}) = \inf_{\phi \in \Iz(\ceil{L/2})} \overline{ G}(\phi,\ceil{L/2}).
\end{align}
By definition $	\Iz(\ceil{L/2}) \subset \Iz(L_0)$, for $L_0\in\mathcal{L}^+\setminus\{\ceil{L/2}\}$, and \eqref{eq:inf:LO-lower} therefore implies that~\eqref{InfLodd} holds (when $L$ is odd). 

We have showed that~\eqref{halfway_0} and~\eqref{InfLodd} holds, 
therefore~\eqref{infCond} holds and this ends the proof of the theorem.
\end{proof}

\section{Proof of Corollary~\ref{cor:1} }
\label{appendix_Cor1}
\begin{proof}
	Let us assume that~\eqref{c_c} holds. Since $f_0(\phi)$ is periodic with period $2\pi/L$ (Lemma~\ref{Lem_f}~property~\ref{Lem_f:p2}) and $|\Iz(L-1)|>2\pi/L$, where these are defined according to~\eqref{fa} and~\eqref{I1}, respectively, then~\eqref{c_c} implies
	\begin{align}
	\sum_{l=0}^{L-1}G_l(\phi)&>	aLG_{\mathrm{SL}}, &\phi \in \Iz(L-1), \nonumber\\
	\sum_{l=0}^{L-2}G_l(\phi)&>aLG_{\mathrm{SL}} - G_{L-1}(\phi) ,& \phi \in \Iz(L-1).
	\end{align}
	where $a=\ceil{L/2}/\big (\floor{L/2} \big) \geq 1$. 
	Based on~\eqref{I0sl}, $G_{L-1}(\phi) \leq G_{\mathrm{SL}}$ when $\phi \in \Iz(L-1)$.  Hence, $aLG_{\mathrm{SL}} - G_{L-1}(\phi)\geq a(L-1)G_{\mathrm{SL}} $ and thus~\eqref{condition} holds since $\Iz(L-1)=\mathcal{J}$. 
\end{proof}

\end{appendices}

\bibliography{ref}

\begin{thebibliography}{10}

\bibitem{AntPlac2007}
S.~{Kaul}, K.~{Ramachandran}, P.~{Shankar}, S.~{Oh}, M.~{Gruteser},
  I.~{Seskar}, and T.~{Nadeem}, ``Effect of antenna placement and diversity on
  vehicular network communications,'' in {\em 2007 4th Annu. IEEE Commun. Soc.
  Conf. on Sensor, Mesh and Ad Hoc Commun. and Netw.}, pp.~112--121, June 2007.

\bibitem{AntPlac_2_2014}
L.~{Ekiz}, A.~{Posselt}, O.~{Klemp}, and C.~F. {Mecklenbr{\"{a}}uker},
  ``Assessment of design methodologies for vehicular 802.11p antenna systems,''
  in {\em 2014 Int. Conf. on Connected Vehicles and Expo (ICCVE)},
  pp.~215--221, Nov. 2014.

\bibitem{hybrid}
I.~{Ahmed}, H.~{Khammari}, A.~{Shahid}, A.~{Musa}, K.~S. {Kim}, E.~{De
  Poorter}, and I.~{Moerman}, ``A survey on hybrid beamforming techniques in
  {5G}: {Architecture} and system model perspectives,'' {\em IEEE Commun.
  Surveys Tuts.}, vol.~20, pp.~3060--3097, Fourthquarter 2018.

\bibitem{hybrid2}
A.~F. {Molisch}, V.~V. {Ratnam}, S.~{Han}, Z.~{Li}, S.~L.~H. {Nguyen}, L.~{Li},
  and K.~{Haneda}, ``Hybrid beamforming for massive {MIMO}: A survey,'' {\em
  IEEE Commun. Mag.}, vol.~55, pp.~134--141, Sep. 2017.

\bibitem{hybrid3}
S.~{Kutty} and D.~{Sen}, ``Beamforming for millimeter wave communications: An
  inclusive survey,'' {\em IEEE Commun. Surveys Tuts.}, vol.~18, pp.~949--973,
  Secondquarter 2016.

\bibitem{config2019}
S.~A. {Busari}, K.~M.~S. {Huq}, S.~{Mumtaz}, J.~{Rodriguez}, Y.~{Fang}, D.~C.
  {Sicker}, S.~{Al-Rubaye}, and A.~{Tsourdos}, ``Generalized hybrid beamforming
  for vehicular connectivity using {THz} massive {MIMO},'' {\em IEEE Trans.
  Veh. Technol}, vol.~68, pp.~8372--8383, Sep. 2019.

\bibitem{AC}
F.~{Gholam}, J.~{Via}, and I.~{Santamaria}, ``Beamforming design for simplified
  analog antenna combining architectures,'' {\em IEEE Trans. Veh. Technol},
  vol.~60, pp.~2373--2378, June 2011.

\bibitem{mainACN}
K.~K. Nagalapur, E.~G. Str{\"{o}}m, F.~Br{\"{a}}nnstr{\"{o}}m, J.~Carlsson, and
  K.~Karlsson, ``Robust connectivity with multiple directional antennas for
  vehicular communications,'' {\em IEEE Trans. Intell. Transp. Syst.}, Dec.
  2019.

\bibitem{patch}
E.~C. Neira, ``Antenna evaluation for vehicular applications in multipath
  environment,'' {\em Licentiate thesis, Dept. of Signals and Syst., Chalmers
  Univ. of Tech., Gothenburg, Sweden}, 2017.

\bibitem{AngSpread2011}
T.~{Abbas}, J.~{Karedal}, F.~{Tufvesson}, A.~{Paier}, L.~{Bernado}, and A.~F.
  {Molisch}, ``Directional analysis of vehicle-to-vehicle propagation
  channels,'' in {\em 2011 IEEE 73rd Veh. Tech. Conf. (VTC Spring)}, May 2011.

\bibitem{karedal2009geometry}
J.~Karedal, F.~Tufvesson, N.~Czink, A.~Paier, C.~Dumard, T.~Zemen, C.~F.
  Mecklenbr{\"{a}}uker, and A.~F. Molisch, ``A geometry-based stochastic {MIMO}
  model for vehicle-to-vehicle communications,'' {\em IEEE Trans. Wireless
  Commun.}, vol.~8, pp.~3646--3657, July 2009.

\bibitem{CAM}
``Intelligent transport systems vehicular communications; basic set of
  applications; part 2: Specification of cooperative awareness basic service,''
  {\em ETSI EN 302 637-2 V1.3.2}, Nov. 2014.

\bibitem{Aoinfo}
S.~{Kaul}, M.~{Gruteser}, V.~{Rai}, and J.~{Kenney}, ``Minimizing age of
  information in vehicular networks,'' in {\em IEEE Commun. Soc. Conf. on
  Sensor, Mesh and Ad Hoc Commun. and Netw.}, pp.~350--358, June 2011.

\bibitem{PIR}
M.~E. Renda, G.~Resta, P.~Santi, F.~Martelli, and A.~Franchini, ``{IEEE
  802.11p} {VANets}: Experimental evaluation of packet inter-reception time,''
  {\em Comput. Commun.}, vol.~75, pp.~26--38, 2016.

\bibitem{PIR0}
T.~Elbatt, S.~Goel, G.~Holland, H.~Krishnan, and J.~Parikh, ``Cooperative
  collision warning using dedicated short range wireless communications,'' in
  {\em Proc. of ACM VANET}, vol.~2006, Sept. 2006.

\bibitem{Twindow}
{Fan Bai} and H.~{Krishnan}, ``Reliability analysis of {DSRC} wireless
  communication for vehicle safety applications,'' in {\em IEEE Intell. Transp.
  Syst. Conf.}, Sept. 2006.

\bibitem{MRC2002}
B.~Holter and G.~{\O}ien, ``The optimal weights of maximum ratio combiner using
  an eigenfilter approach,'' in {\em Proc. 5th IEEE Nordic Signal Process.
  Symp. (NORSIG 2002)}, Oct. 2002.

\bibitem{adp18}
``Mathematical models for radiodetermination radar systems antenna patterns for
  use in interference analyses,'' {\em Recommendation ITU-R M1851-1}, Jan.
  2018.

\end{thebibliography}
\bibliographystyle{ieeetr}

\end{document}